\documentclass[11pt]{article}
\usepackage{odonnell,marginnote}
\newcommand{\rnote}[1]{}
\newcommand{\knote}[1]{}
\newcommand{\CC}[2]{\mathcal{C}\mathcal{C}(#1;#2)}
\newcommand{\Esc}[2]{\ol{\mathrm{G}}(#1,#2)}
\newcommand{\Gu}[2]{\ol{\mathrm{G}}(#1,#2)}
\newcommand{\Gad}[2]{\mathrm{G}(#1,#2)}
\newcommand{\G}[2]{\mathrm{G}(#1,#2)}
\newcommand{\Ggeneric}[2]{\mathrm{G}(#1,#2)}
\newcommand{\Center}[2]{\mathrm{Z}(#1,#2)}
\newcommand{\Const}[4]{C^{#1,#2}_{#3 #4}}
\newcommand{\hgt}{{\mathrm{ht}}}
\newcommand{\rank}{{\mathrm{rank}}}
\newcommand{\SL}{{\mathrm{SL}}}
\newcommand{\SP}{{\mathrm{Sp}}}
\newcommand{\PSp}{{\mathrm{PSp}}}
\newcommand{\PSL}{{\mathrm{PSL}}}

\newcommand{\Exp}{{\mathrm{Exp}}}

\newcommand{\Cay}{\mathrm{Cay}}
\newcommand{\Link}{\mathrm{Link}}
\newcommand{\complex}{\mathfrak{K}}
\usepackage{etoolbox}
\usepackage{dynkin-diagrams}

\def\row#1/#2!{#1_{\IfStrEq{#2}{}{n}{#2}} & \dynkin{#1}{#2}\\}

\title{High-Dimensional Expanders from Chevalley Groups}
\author{Ryan O'Donnell\footnotemark[1]${~}^{\tiny{\textcircled{r}}}$ \and Kevin Pratt%
\thanks{Computer Science Department, CMU. \{\texttt{odonnell},\texttt{kpratt}\}\texttt{cs.cmu.edu}.
\quad {\tiny \textcircled{r}} The author ordering was randomized.}${~}^{\tiny{\textcircled{r}}}$}
\date{}

\begin{document}

\maketitle

\begin{abstract}
Let $\Phi$ be an irreducible root system (other than $G_2$) of rank at least $2$, let $\F$ be a finite field with $p = \chr \F > 3$, and let $\Ggeneric{\Phi}{\F}$ be the corresponding Chevalley group.
We describe a strongly explicit high-dimensional expander (HDX) family of dimension~$\rank(\Phi)$, where $\Ggeneric{\Phi}{\F}$ acts simply transitively on the top-dimensional faces; these are $\lambda$-spectral HDXs with $\lambda \to 0$ as $p \to \infty$.
This generalizes a construction of Kaufman and Oppenheim (STOC 2018), which corresponds to the case $\Phi = A_d$.
Our work gives three new families of spectral HDXs of any dimension $\ge 2$, and four exceptional constructions of dimension $4$,~$6$,~$7$, and $8$.
\end{abstract}

\section{Introduction}
In 1989, Babai, Kantor, and Lubotzky made a conjecture that significantly guided research on expander graphs:
\begin{conjecture}      \label{conj:1}
    (\cite{BKL89}.) There are constants $k \in \N$ and $\lambda < 1$ such that for every nonabelian finite simple group~$G$, there is a symmetric set $S \subset G$ of $2k$ generators such that the Cayley graph $\Cay(G,S)$ is a $\lambda$-spectral expander graph.
\end{conjecture}
\noindent (Here we say that a graph is a $\lambda$-spectral expander if all the eigenvalues of its random walk matrix, excluding the largest, are at most~$\lambda$.)

Notable achievements toward the conjecture include: Kassabov's proof~\cite{Kas07} for the alternating groups; work of Lubotzky and Nikolov~\cite{KLN06} proving the conjecture for non-Suzuki groups of Lie type (the Chevalley groups and their twisted versions); and, the Breuillard--Green--Tao~\cite{BGT11} proof for the Suzuki groups.
In light of the Classification of Finite Simple Groups~\cite{Asc04}, these completed the proof of \Cref{conj:1}.
An immediate consequence is that for every nonabelian simple group~$G$, there is a $2k$-regular $\lambda$-spectral expander~$\complex$ such that $G$ acts transitively on the vertices of~$\complex$.

Having expander graphs with such nontrivial symmetry properties (or even stronger ones) has played an important role in applications to computer science.
For example, motivated by the search for locally testable codes (see~\cite{KS08}), Kaufman and Wigderson~\cite{KW16} made substantial progress on finding so-called ``highly symmetric'' LDPC codes of constant rate and relative distance (``good'') using expanding Cayley graphs of nonabelian groups; at the same time, they showed that highly symmetric LDPC codes arising from abelian --- or even solvable --- groups cannot work.
Later, notable work of Kaufman and Lubotzky~\cite{KL12} (see also~\cite{Bec16}) positively resolved the problem, giving explicit, highly symmetric, good LDPC codes; the main tool was the use of explicit \emph{edge-transitive} (not just vertex-transitive), highly expanding (indeed, Ramanujan) Cayley graphs of $\PSL_2(\F_q)$ (for $q = 4093$).
In turn, the existence of these highly-symmetric expanders arose from the construction of Ramanujan \emph{high-dimensional expanders} (HDXs)~\cite{Bal00,CSZ03,Li04,LSV05,LSV05b,Sar04} from Bruhat--Tits buildings, relying on the Lafforgue's work~\cite{Laf02} on the Langlands correspondence.

High-dimensional expanders --- defined, say, as simplicial complexes where the $1$-skeleton of every link is a $\lambda$-spectral expander --- have been crucial in many new works in theoretical computer science, either through inspiration, their spectral analysis, or their direct construction. 
Example applications include results in analysis of Boolean functions~\cite{DDFH18}, computational geometry~\cite{FGLNP12}, inapproximability~\cite{AJT20,DFHT21}, list-decoding~\cite{AJQST20,DHKNT21}, Markov chain mixing~\cite{AL20}, property testing~\cite{DK17,DD19,KM20,KO20}, and quantum codes~\cite{EKZ20,KT21}; particularly notable examples including the resolution of the Mihail--Vazirani Conjecture on the bases-exchange walk for matroids~\cite{ALOV19} and the construction of locally testable codes of constant rate, distance, and locality~\cite{DELLM21,PK21}.


\subsection{Our goal}   \label{sec:goal}
In this paper, we investigate a problem similar to \Cref{conj:1} for high-dimensional expanders.
Namely, for nonabelian finite simple groups~$G$, we seek:
\begin{enumerate}
    \item bounded-degree $\lambda$-spectral HDXs whose top-dimensional faces are acted on transitively by~$G$,
    \item with $\lambda$ arbitrarily close to~$0$, as opposed to merely bounded away from~$1$.
\end{enumerate}
(Recall that existence of highly symmetric good LDPC codes was resolved by obtaining \emph{one}-dimensional HDXs --- i.e., expander graphs --- with both properties.)
The aforementioned HDXs built from Bruhat--Tits buildings~\cite{Li04,LSV05,LSV05b,Sar04} have property~(2) above, and the work of Kaufman and Lubotzky~\cite{KL12} also verified property~(1) for $G = \PSL_3(\F)$ (for  $\chr \F$ sufficiently large).
Later, Kaufman and Oppenheim~\cite{KO18} gave a new (and elementary) construction of HDX families of any dimension~$d$ satisfying both (1)~and~(2) with $G = \PSL_{d+1}(\F)$.
These two constructions are the only previous examples of bounded-degree $\lambda$-spectral HDXs of which we are aware.  To quote the final remark from~\cite{LSV05}: ``\emph{Of course one hopes eventually to define and construct Ramanujan complexes as quotients of the Bruhat--Tits buildings of other simple groups as well.}''

\paragraph{Results.}
We give strongly explicit constructions of $d$-dimensional HDX families satisfying properties (1)~and~(2) above, for any rank-$d$ \emph{Chevalley group}~$G$ (except for ``$G_2$'') over any field~$\F$ of characteristic exceeding~$3$.\footnote{In fact, we can show that our construction works for characteristic~$3$ when one excludes the case of~$G_2$.  But for simplicity of presentation we will just assume the characteristic exceeds~$3$.\label{foot:5}}
Informally, Chevalley groups (also known as the untwisted groups of Lie type) are the finite-field analogues of continuous Lie groups.
These groups are specified by two pieces of data: a \emph{root system} $\Phi$, consisting of a set of vectors in~$\R^d$ with certain symmetry properties, and a finite field $\F$.
Our work gives a general recipe that produces HDX families from Chevalley groups with $\Phi$~and the characteristic of~$\F$ being fixed, and with $|\F|$ growing.
Our approach generalizes that of~\cite{KO18}, which corresponds to the case $\Phi = A_d$.
As with their work, our construction incidentally gives new families of strongly explicit $\Delta$-degree-bounded $\lambda$-spectral expander graphs, with $\lambda \to 0$ as $\Delta \to \infty$.\rnote{Some comments about graphs of constant link?}

\subsection{Our approach}
As in~\cite{KO18}, we associate to $G$ a \emph{coset complex}, a kind of $d$-dimensional simplicial complex determined by~$G$ and a choice of subgroups $H_1, \ldots, H_{d+1}$ of~$G$.
A few challenges arise in generalizing the construction of~\cite{KO18} to Chevalley groups $G$ of type other than~$A_d$.
One immediate question is: what is a ``good" choice of $H_1, \ldots, H_{d+1}$?
We give one such choice, which has an elegant description in terms of the root system~$\Phi$ associated with $G$: the $H_i$'s are certain unipotent subgroups of~$G$ (these are essentially groups of upper unitriangular matrices), obtained from a set of fundamental roots of $\Phi$. While all of our constructions can be realized with matrices (see the examples in the next section), it is more convenient in our analysis to work with a set of generators and relations of $G$ known as the \emph{Steinberg presentation}. In particular, the \emph{Chevalley commutator formula} gives us workable descriptions of the subgroups~$H_i$, and the links of our complexes.

As in~\cite{KO18}, we apply the \emph{trickling down} theorem of \cite{Opp18} (originating in work of Garland~\cite{Gar73}) to show that these coset complexes have expanding links.
This theorem says that under a mild connectivity condition, it suffices to show that the links of the $(d-2)$-dimensional faces are good expander graphs.
The connectivity condition will follow from some calculations using the properties of Chevalley groups and root systems. In their case of $\Phi = A_d$, Kaufman--Oppenheim establish expansion of links by appealing to a general result of Ershov--Jaikin-Zapirain~\cite{EJ10} on expansion in certain groups of nilpotency class two.
Unfortunately, to handle root systems~$\Phi$ that are not ``simply-laced'', one would need an analogous result for groups of nilpotency class three (and higher, when $\Phi = G_2$).
Related results were given in \cite{EJK17} (see its Sec.~10.3), but these are not strong enough for our setting.
An alternative, and much simpler, proof of expansion of the Kaufman--Oppenheim complexes was given by Harsha and Saptharishi~\cite{HS19}; their proof was quite specific to the $\Phi = A_d$ case, but we were much inspired its elementary nature.

We prove expansion by observing that, when $\Phi \neq G_2$, the squares of the links of the $(d-2)$-dimensional faces are Cayley graphs of abelian groups. This allows us to express their eigenvalues as character sums, which we bound with an elementary argument that ultimately boils down to the Schwartz--Zippel lemma.  (In the $G_2$ case, the squared links are not abelian Cayley graphs, but we discuss some approach that might be used ito show their expansion.)

\subsection{Example constructions}  \label{sec:eg}
In this section we explicitly give the easiest new HDX family implied by our work.
We start by recalling the basic construction of \cite{KO18}, arising from the group $G = \SL_3(\F)$.\footnote{In \cite{KO18} they work over the ring $\F_p[x]/(x^m)$ rather than the field~$\F_{p^m}$, but this does not materially change their result, and we prefer to work with the field.  Also, regarding the distinction between $\SL_3(\F)$ and $\PSL_3(\F)$, see \Cref{foot:universal}.}
Let $\F = \F_p[x]/( f )$ where $p$ is prime and $f \in \F_p[x]$ is irreducible of degree~$m$.
Now define the following three subgroups of $G$:
\begin{align*}
H_1 &= \left \{\begin{bmatrix}
1 & \ell_1 & Q \\
 & 1 & \ell_2 \\
 &  & 1\\
\end{bmatrix} : \deg(\ell_1) , \deg(\ell_2) \le 1, \deg(Q) \le 2 \right \},\\
H_2 &= \left \{ \begin{bmatrix}
1 & \ell_1 &  \\
 & 1 &  \\
\ell_2 & Q & 1\\
\end{bmatrix}: \deg(\ell_1) , \deg(\ell_2) \le 1, \deg(Q) \le 2 \right \},\\
H_3 &= \left \{ \begin{bmatrix}
1 &  & \\
Q & 1 & \ell_1 \\
\ell_2 &  & 1\\
\end{bmatrix} : \deg(\ell_1) , \deg(\ell_2) \le 1, \deg(Q) \le 2 \right \}.
\end{align*}
Let $\complex(p,m)$ be the $2$-dimensional simplicial complex whose vertices are the cosets of these subgroups inside $\SL_3(\F)$, and where a ``triangle'' ($2$-dimensional face) is added between a triple of cosets $g_1 H_1,g_2 H_2,g_3 H_3$ whenever $g_1 H_1 \cap g_2 H_2 \cap g_3 H_3 \neq \emptyset$. Edges are included between any two cosets contained in a common triangle. This is an example of a \emph{coset complex}, a well-studied construction from the theory of algebraic groups.

In \cite{KO18} it was shown that for any fixed $\lambda > 0$, and for sufficiently large $p$, the complex $\complex(p,m)$ is a bounded-degree $\lambda$-spectral HDX.
(Here ``bounded-degree'' means that each vertex of $\complex(p,m)$ is contained in a number of triangles depending only on~$p$.)
Moreover, $\SL_3(\F)$ acts simply transitively on the set of triangles in $\complex(p,m)$.
In a similar manner, Kaufman and Oppenheim show how a $d$-dimensional HDX family can be associated to~$\SL_{d+1}(\F)$.

Following this, the most basic new construction provided by our work is as follows.
Again, we form a coset complex, but this time we will consider cosets of subgroups of the $4 \times 4$ \emph{symplectic group}, $\SP_4(\F)$,\footnote{%
Technically, this group is not simple; it only becomes the simple group $\PSp_4(\F)$ upon identifying the matrices $A$~and~$-A$.
This is an example of the (very minor) distinction between ``universal'' and ``adjoint'' Chevalley groups that is explained in \Cref{def:adj}.\label{foot:universal}} defined by
\[\SP_4(\F) = \left \{ A \in \F^{4 \times 4} : A\begin{bmatrix}
0 & I_{2\times 2}\\
-I_{2\times 2} & 0
\end{bmatrix}A^\intercal = \begin{bmatrix}
0 & I_{2\times 2}\\
-I_{2\times 2} & 0
\end{bmatrix} \right \} .\]
The vertices of our coset complex $\complex_{\SP_4}(p,m)$ will be the cosets of the following subgroups of $\SP_4(\F)$:
\begin{align*}
H_1 &=
\left \{ \begin{bmatrix}
1 & \ell_1 & C & \ell_1 \ell_2 + Q\\
 & 1 & Q & \ell_2\\
 &  & 1& \\
 &  & -\ell_1 & 1
\end{bmatrix} : \deg(\ell_1),\deg(\ell_2) \le 1, \deg(Q) \le 2, \deg(C) \le 3 \right \},\\
H_2 &=
\left \{ \begin{bmatrix}
1 & \ell_1 &  & \\
 & 1 &  & \\
\ell_2 & Q & 1& \\
\ell_1 \ell_2 + Q & C & -\ell_1 & 1
\end{bmatrix}  : \deg(\ell_1),\deg(\ell_2) \le 1, \deg(Q) \le 2, \deg(C) \le 3 \right \},\\
H_3 &= \left \{
\begin{bmatrix}
1 &  &  & \\
 & 1 &  & \ell_1\\
\ell_2 &  & 1& \\
 &   &  & 1
\end{bmatrix} : \deg(\ell_1), \deg(\ell_2) \le 1 \right \}.
\end{align*}
The triangles in $\complex_{\SP_4}(p,m)$ are again added between triples of cosets whenever they have a nontrivial intersection.
Our work shows that for any $\lambda > 0$, provided~$p \geq 2\frac{(1+\lambda)^2}{\lambda^2}$, the $2$-dimensional complexes $(\complex_{\SP_4}(p,m))_{m}$ form a (strongly explicit) $\lambda$-spectral HDX family of size $\Theta(p^{10m-4})$ in which each vertex participates in at most $p^{22}$ triangles.
Moreover, the group $\SP_4(\F_{p^m})$ acts on $\complex_{\SP_4}(p,m)$, with the action being transitive on triangles.
Finally, we remark that the underlying skeleton of $\complex_{\SP_4}(p,m)$ is a (strongly explicit) $\lambda$-spectral expander graph of degree at most $p^{11}$ and with $\Theta(p^{10m-4})$ vertices.
Since this graph is tripartite, its smallest eigenvalue is at least $-1/2$, and it is therefore also a \emph{two}-sided $1/2$-spectral expander.

\subsection{Outline}
In \Cref{sec:prelim} we give an overview of high-dimensional spectral expansion and coset complexes. We then briefly discuss Chevalley groups and root systems, making explicit all facts about Chevalley groups that we will need.

In \Cref{sec:construct} we give the choice of subgroups used in our coset complex construction. We show in \Cref{cor:connectedlinks} that these have the connectivity properties needed to apply the trickling down \Cref{thm:trickle}. In \Cref{sec:explink} we show that the links of the $(d-2)$-dimensional faces in these complexes are good expander graphs. By \Cref{fact:link}, this conveniently reduces to studying the expansion of vertex links in three different $2$-dimensional complexes, two of which are the examples in \Cref{sec:eg}.

We conclude with further questions. It is interesting to ask if our analogue of \Cref{conj:1} has an affirmative answer when~$G$ is a more ``combinatorial" group; for example, the symmetric/alternating group. We also leave open the case of the Chevalley group based on root system~$G_2$; we conjecture it has the desired expansion properties, and suggest an approach to proving this.

\section{Preliminaries}\label{sec:prelim}
Let $\N = \{0,1,2,\ldots\}$, $\Z_+ = \{1,2,3,\ldots\}$. We identify elements in a finite field $\F$ of size $p^m$ (where $p$~is prime) with polynomials in $\F_p[x]/(f)$ for some irreducible polynomial $f \in \F_p[x]$ of degree~$m$. When we write $\deg(t)$ for $t \in \F$ we mean the degree of the corresponding polynomial in the quotient ring.
If $g$ and $h$ are elements of a group, we use the notation $[g,h] = g^{-1}h^{-1}gh$ for their commutator. 

\subsection{High-dimensional spectral expansion}
In this section we recall the notion of spectral HDX families. Let $\complex(0)$ be a finite set. A \emph{simplicial complex} $\complex$ with vertex set $\complex(0)$ is a collection of subsets of $\complex(0)$ satisfying the following conditions:
\begin{enumerate}
\item $\{v\} \in \complex$ for all $v \in \complex(0)$;
\item If $\sigma \in \complex$, then $\tau \in \complex$ for all $\tau \subseteq \sigma$.
\end{enumerate}
Said differently, $\complex$ is a downward-closed hypergraph on the set $\complex(0)$.  For $i = -1, 0, 1, 2, \dots$, we denote by $\complex(i)$ the set of subsets of size $i+1$ in $\complex$. An element of $\complex(i)$ is called an \emph{$i$-dimensional face}. $\complex$~is said to be \emph{pure} if all maximal faces are $d$-dimensional for some~$d$; in this case we say that $d$~is the \emph{dimension} of $\complex$, denoted $\dim \complex$. (In this work, all simplicial complexes will be pure.)
Note that a $1$-dimensional simplicial complex can be identified with an ordinary graph.
We say that $\complex$ is of \emph{$\Delta$-bounded degree} if every vertex participates in at most~$\Delta$ maximal faces; and, we say that $\complex$ is \emph{$k$-partite} if there is a partition of $\complex(0)$ into $k$ parts such that each face has intersection size at most~$1$ with each part. (Pure $(d+1)$-partite complexes are sometimes called \emph{balanced}, or \emph{numbered}.)

The \emph{link} of a face $\sigma \in \complex$ is the simplicial complex $\Link_\sigma(\complex) = \{\tau \setminus \sigma: \tau \in \complex, \sigma \subseteq \tau\}$. In particular, the link of the $(-1)$-dimensional face $\emptyset$ is $\complex$.
For a pure $d$-dimensional complex~$\complex$, we define the \emph{$1$-skeleton} of~$\complex$ to be the \emph{multigraph} on vertex set~$\complex(0)$ in which $j,k \in \complex(0)$ are connected by a number of edges equal to the number of $d$-dimensional faces containing $\{j,k\}$.  We will say that $\complex$ is \emph{connected} if its $1$-skeleton is a connected (multi)graph.  Finally, for a face~$\sigma$, we introduce the notation $K_\sigma$ for the $1$-skeleton of $\Link_\sigma(\complex)$, and we will write $\lambda_2(K_\sigma)$ for the second largest eigenvalue of the standard random walk matrix of~$K_\sigma$. (This refers to the walk on the vertices of~$K_\sigma$ in which a random out-edge is taken at each step.)\\

By now, the most common definition of expansion for HDXs is probably the following:
\begin{definition}
    (\cite{Opp18,KO18}.) A $d$-dimensional pure simplicial complex $\complex$ is a \emph{$\lambda$-spectral HDX} (also known as \emph{$\lambda$-link} or \emph{$\lambda$-local-spectral} HDX) if $\lambda_2(K_\sigma) \le \lambda$ for all faces $\sigma$ of dimension at most~$d-2$.
\end{definition}
\noindent (Note that the $d = 1$ case yields the usual notion of a $\lambda$-expander graph, one in which the second eigenvalue of the random walk matrix is at most~$\lambda$.)  The trickling down theorem~\cite{Opp18} essentially shows that a $d$-dimensional complex is an HDX provided the links of its $(d-2)$-dimensional faces are $\lambda$-expander graphs for $\lambda < \frac{1}{d}$:
\begin{theorem}\label{thm:trickle}
    \textnormal{(\cite{Opp18}.)}
     Let $\complex$ be a $d$-dimensional pure simplicial complex in which $\Link_\sigma(\complex)$ is connected for all $\sigma \in \complex(i)$, $i \leq d-2$.  Further suppose that $\lambda_2(K_\sigma) \leq \gamma \leq \frac{1}{d}$ for all $\sigma \in \complex(d-2)$.
    Then $\complex$ is a $\left (\frac{\gamma}{1-(d-1)\gamma} \right)$-spectral HDX.
\end{theorem}
\noindent (We remark that in the case $d = 2$, if $\complex$ is a Cayley graph then the conclusion of this theorem can be improved by a factor of~$2/\sqrt{3}$; see~\cite{Zuk03}.)\\

The objects we seek are (highly symmetric versions of) the following:
\begin{definition}
    A \emph{$d$-dimensional, $\Delta$-bounded degree, $\lambda$-spectral HDX family} is a sequence $(\complex_n)_{n \in \N}$ of pure $d$-dimensional, $\Delta$-bounded degree complexes, with $\complex_n$ having some $n' = \Theta(n)$ vertices, such that $\complex_n$ is a $\lambda$-spectral HDX for sufficiently large~$n$.
    We also say the family is \emph{explicit} if there is a $\poly(n)$-time algorithm for computing the description of~$\complex_n$, and \emph{strongly explicit} if there is a $\polylog(n)$-time algorithm.  (See the proof of \Cref{thm:main}, item~1 for more details.)
\end{definition}

\subsection{Coset complexes}
The following notion has been studied since at least the 1950 PhD thesis of Lann\'{e}r~\cite{Lan50}:
\begin{definition}
    Let $G$ be a finite group and let $\calH = (H_1, \dots, H_{d+1})$ be a sequence of subgroups.
    The associated \emph{coset complex} $\CC{G}{\calH}$ is the pure $d$-dimensional, $(d+1)$-partite simplicial complex with vertices being the cosets $\bigsqcup_i G/H_i$, and with maximal faces $\{gH_1, \dots, gH_{d+1} : g \in G\}$.
    Equivalently, a set of cosets forms a face if all cosets have an element in common.
\end{definition}
Some well-studied instances of coset complexes are \emph{Coxeter complexes} and \emph{Tits buildings}~\cite{Bjo84}.
\begin{definition}                                        \label{fact:types}
    The $i$th part of the $(d+1)$-partite coset complex $\CC{G}{\calH}$ is the coset~$G/H_i$, and the \emph{type} of a face~$\sigma$ refers to the subset of parts~$[d+1]$ to which its vertices belong.
\end{definition}
The group $G$ naturally acts on $\CC{G}{\calH}$ by left-multiplication, and it is easy to see the following:
\begin{fact}                                        \label{fact:type-preserving}
    The action of~$G$ on the $(d+1)$-partite complex $\CC{G}{\calH}$ is \emph{type-preserving} (it does not change the type of any face), and transitive on the maximal faces.
    Moreover, the action is simply transitive if $H_1 \cap H_2 \cap \cdots \cap H_{d+1} = \{1\}$.
\end{fact}
\noindent (In fact, Lann\'{e}r~\cite{Lan50} showed that whenever there is a $G$-action on some $(d+1)$-partite complex that is type-preserving and transitive on maximal faces, then the complex must be of the form $\CC{G}{\calH}$ for some subgroups $H_1, \dots, H_{d+1}$.)

We can also easily understand the connectivity and link structure of coset complexes, as the following facts show.

\begin{definition}
    Given $\calH$ and $T \subseteq [d+1]$ we write $H_T = \bigcap_{i \in T} H_i$, with the convention that $H_\emptyset = \la H_1, \dots, H_{d+1} \ra$, the subgroup of~$G$ generated by~$\calH$.
\end{definition}
The following facts are easy to prove:
\begin{fact}                                        \label{fact:connected}
    (\cite{AH93}.)
    $\CC{G}{\calH}$ is connected  if and only if $H_\emptyset = G$.
\end{fact}
\begin{fact}                                        \label{fact:link}
    (\cite[p.~13]{Gar79}, \cite{HS19}.)
    Let $\sigma$ be a face in $\CC{G}{\calH}$ of type~$T \neq \emptyset$.
    Then the link of~$F$ is isomorphic to the coset complex $\CC{H_T}{(H_{T \cup \{i\}} : i \not \in T)}$.
\end{fact}
Note that \Cref{fact:link} says that, up to isomorphism, the link of a face only depends on its type. This will help us apply \Cref{thm:trickle}, as we will only have to consider a small number of cases.
Finally we quote another easy-to-prove fact from Kaufman and Oppenheim, which we can use to pass between the (very slightly different) different universal and adjoint Chevalley groups:
\begin{fact}                                     \label{fact:ko2.4}
    (\cite[essentially Prop.~2.12]{KO18}.)
    Let $\overline{\complex} = \CC{\ol{G}}{\ol{\calH}}$ be a coset complex with $\ol{\calH} = (\ol{H}_1, \dots, \ol{H}_{d+1})$, suppose $Z \triangleleft \ol{G}$ is a normal subgroup (e.g., if $Z$ is the center of~$\ol{G}$), and suppose that $Z \cap \ol{H}_i = \{1\}$ for all $i \in [d+1]$.
    Then for $G = \ol{G} / Z$ and $\calH= (H_1, \dots, H_{d+1})$, where $H_i = \ol{H}_i Z / Z$, the coset complex $\complex = \CC{G}{\calH}$ is ``covered'' by $\overline{\complex}$, and the following property holds: every link in~$\complex$ of type $T \neq \emptyset$ is isomorphic to every link of type~$T$ in $\ol{\complex}$
\end{fact}
\begin{remark}  \label{rem:get-to-adjoint}
    A consequence of \Cref{fact:ko2.4} is that if $\overline{\complex}$ is a $\Delta$-bounded degree, $\lambda$-spectral HDX, then so too is $\complex$; moreover, provided $H_1 Z \cap H_2 Z \cap \cdots \cap H_{d+1} Z = Z$, the group $G$ acts simply transitively on the maximal faces of~$\complex$. We note that our complexes satisfy this condition in \Cref{obs:centerint}.
\end{remark}

\subsection{Root systems}

Killing and Cartan~\cite{Car94} classified simple Lie algebras over~$\C$ via root systems:
\begin{definition}
    A \emph{(reduced) root system} of rank~$d$ is a finite set~$\Phi$ of nonzero vectors spanning a $d$-dimensional real vector space such that for each $\alpha \in \Phi$:
    \begin{itemize}
        \item $\Phi$ is closed under~$w_\alpha$, where $w_\alpha$ is the reflection through the hyperplane orthogonal to~$\alpha$;
        \item $w_\alpha(\beta) - \beta$ is an integer multiple of~$\alpha$ for all $\beta \in \Phi$;
        \item for $\lambda \in \R$ we have $\lambda \alpha \in \Phi$ (if and) only if $\lambda \in \{\pm 1\}$.
    \end{itemize}
    The root system $\Phi$ is \emph{irreducible} if it cannot be written as $\Phi_1 \sqcup \Phi_2$ with $\Phi_1, \Phi_2$ nonempty and lying in orthogonal subspaces.
    Root system $\Phi'$ is said to be \emph{isomorphic} to~$\Phi$ if there is bijection between them that preserves inner products up to a fixed positive scalar multiple.
\end{definition}

\Cref{fig:root-eg} shows the three non-isomorphic rank-$2$ root systems (all of which are irreducible).
The irreducible root systems have been completely classified:
\begin{notation}
    Up to isomorphism, the irreducible root systems are classified as the families $A_d$~($d \geq 1$), $B_d$ ($d \geq 2$), $C_d$ ($d \geq 3$), $D_d$ ($d \geq 4$), and the exceptional systems $G_2$, $F_4$, $E_6$, $E_7$,~$E_8$.
    In all cases, the subscript gives the dimension of the root system.
    For explicit descriptions of these root systems, 
    see e.g.~\cite[Sec.~3.6]{Car89}.
\end{notation}

\begin{figure}[H]
\centering

  \begin{tikzpicture}
    \foreach\ang in {30,90,...,330}{
     \draw[->,blue!80!black,thick] (0,0) -- (\ang:3cm);
    }
    \node[anchor=south west,scale=0.6] at (2.5,-1) {$\alpha$};
    \node[anchor=north east,scale=0.6] at (0,3) {$\beta$};
    \node[anchor=north,scale=0.6] at (0,-3.2) {\huge $A_{2}$};
  \end{tikzpicture}
  \begin{tikzpicture}
    \foreach\ang in {90,180,...,360}{
     \draw[->,blue!80!black,thick] (0,0) -- (\ang:2cm);
    }
    \foreach\ang in {45,135,225,315}{
     \draw[->,blue!80!black,thick] (0,0) -- (\ang:2.828427cm);
    }
    \node[anchor=south west,scale=0.6] at (2,0) {$\alpha$};
    \node[anchor=north east,scale=0.6] at (-2.14,2) {$\beta$};
    \node[anchor=north,scale=0.6] at (0,-3.2) {\huge $B_{2} \cong C_2$};
  \end{tikzpicture}
\begin{tikzpicture}
    \foreach\ang in {60,120,...,360}{
     \draw[->,blue!80!black,thick] (0,0) -- (\ang:1.73205cm);
    }
    \foreach\ang in {30,90,...,330}{
     \draw[->,blue!80!black,thick] (0,0) -- (\ang:3cm);
    }
    \node[anchor=south west,scale=0.6] at (2,0) {$\alpha$};
    \node[anchor=north east,scale=0.6] at (-2.5,2) {$\beta$};
    \node[anchor=north,scale=0.6] at (0,-3.2) {\huge $G_{2}$};
  \end{tikzpicture}
\caption{The rank $2$ (irreducible) root systems, with a simple set $\{\alpha, \beta\}$ shown.}     \label{fig:root-eg}
\end{figure}
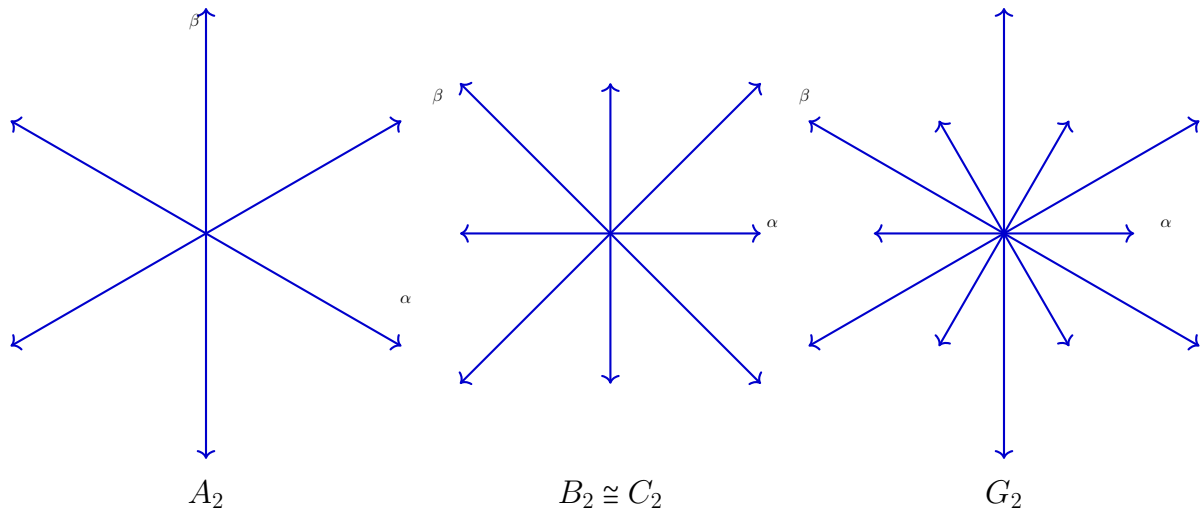

\begin{remark}  \label{rem:2rootspan}
    The restriction of a root system to a subspace is also a root system.
    Thus if $\Phi$ is a root system containing roots $\alpha$, $\beta$, and $\alpha + \beta$, then the restriction of $\Phi$ to the subspace spanned by $\alpha$ and~$\beta$ must be (isomorphic to) $A_2$, $B_2$, or~$G_2$.
    In fact, since $G_2$ is the only irreducible root system containing vectors at an angle of~$30^\circ$ (see, e.g.,~\cite[Sec.~3.6]{Car89}), an irreducible root system containing~$G_2$ as a subsystem must in fact be isomorphic to~$G_2$.
\end{remark}
%
%

Let us now record a handy fact involving the inner product $\alpha \cdot \beta$ of two roots:
\begin{fact}    \label{rootsum}
    (\cite[p.~45, Lem.~9.4]{Hum72}.)
    Let $\alpha, \beta$ be roots.  If $\alpha \cdot \beta < 0$ then $\alpha+\beta \in \Phi \cup \{0\}$, and if $\alpha \cdot \beta > 0$ then $\alpha-\beta \in \Phi \cup \{0\}$.
\end{fact}
This fact can be used to prove another simple result (which is surely well known, though we could not find a reference):
\begin{fact}    \label{fact:rootdecomp}
    Let $\Phi$ be an irreducible root system of rank at least~$2$, and let $\alpha \in \Phi$.  Then $\alpha$ is the sum of two other roots.
\end{fact}
\begin{proof}
    We claim there must exist a root $\beta \neq \pm \alpha$ with $\alpha \cdot \beta \neq 0$.
    Otherwise, every root is either orthogonal to~$\alpha$ or parallel to~$\alpha$, meaning $\Phi$ is either irreducible or of rank~$1$.
    We may assume $\alpha \cdot \beta > 0$, by replacing $\beta$ by the root $-\beta$, if necessary.
    Thus \Cref{rootsum} tells us that $\alpha - \beta \in \Phi$.
    But now $\alpha - \beta$ and $\beta$ are roots summing to~$\alpha$.
\end{proof}

We now discuss ``simple'' subsets of roots:
\begin{definition}  \label{def:simple}
    Let $\Phi$ be a root system spanning $\R^d$.
    A set of roots $\Pi = \{\alpha_1, \dots, \alpha_d\} \subseteq \Phi$ is called \emph{simple} (or a \emph{base}) if it is a basis for~$\R^d$, and
    every root $\gamma \in \Phi$ may be expressed as
    \[
        \gamma = n_1 \alpha_1 + \cdots + n_d \alpha_d
    \]
    either with $n_1, \dots, n_d \in \N$ or with $-n_1, \dots, -n_d \in \N$.
    (Since $\Pi$ is a basis, there is a unique such expression.)
    In the former case, $\gamma$ is called a \emph{positive root}; in the latter case,  a \emph{negative root}.
    One also defines the \emph{height} of~$\gamma$ (with respect to~$\Pi$, or more generally a set of linearly independent roots whose span contains $\gamma$), denoted $\hgt_\Pi(\gamma)$, to be $\sum_{i=1}^d \abs{n_i}$.
\end{definition}

\noindent (In \Cref{fig:root-eg}, each root system has labeled a simple set $\{\alpha,\beta\}$.)

In a certain sense, up to symmetries there is a unique choice of simple roots for a given root system:
\begin{fact} (\cite[Prop.~2.1.2, Cor~2.2.5]{Car89}.)
    Every root system has a set of simple roots.
    Further, for any two simple sets, there is a unique reflection $w_\alpha$ mapping one to the other.
%
\end{fact}

\begin{definition}
For any subset $\Psi \subseteq \Phi$, we write $\Psi^+ = \Phi \cap \{\sum_{\alpha \in \Psi} n_\alpha \alpha : n_\alpha \in \N\}$, and $\Psi^{-} = -\Psi^+$.
\end{definition}

\begin{fact}\label{fact:simplechoice}
Let $\Psi \subset \Phi$ be a set of linearly independent roots. Then there is set of simple roots $\Pi$ of~$\Phi$ where $\Psi \subseteq \Pi^+$.
\end{fact}
\begin{proof}
We can always find a hyperplane $H$ not containing any root, and where all of $\Psi$ is contained on one side of $H$. Then by \cite[Thm.~8.16]{Hal03}, there is a set of simple roots $\Pi$ such that the roots in~$\Phi$ on this side of $H$ are positive with respect to $\Pi$.
\end{proof}

The following fact is very similar to a standard one about root systems, but it is usually only stated when $\{\alpha, \dots, \alpha_m\}$ form a simple set (see, e.g., \cite[Lem.~3.6.2]{Car89}):
\begin{fact}    \label{fact:prefixsum}
    Let $A = \{\alpha_1, \dots, \alpha_\ell\} \subseteq \Phi$ be \emph{any} set of roots, and suppose that $\gamma = \sum_{i=1}^\ell n_i \alpha_i \in \Phi$ for $n_1, \dots, n_\ell \in \N$.
    Then we may express $\gamma = \sum_{j=1}^m \alpha_{i_j}$ for certain $i_1, \dots, i_m \in [\ell]$ in such a way that all the prefix-sums $\sum_{j=1}^k \alpha_{i_j}$ ($1 \leq k \leq \ell$) are in~$\Phi$.
\end{fact}
\begin{proof}
    By induction, it suffices to show that if $\gamma$ is not already in~$A$, then there exists $i_0 \in [\ell]$ with $n_{i_0} > 0$ such that $\gamma - \alpha_{i_0} \in \Phi$.
    To do this, note that $0 < \gamma \cdot \gamma = \sum_{i=1}^\ell  n_i (\gamma \cdot \alpha_i)$, and since the $n_i$'s are nonnegative we must have $\gamma \cdot \alpha_{i_0} > 0$ for (at least) one~$i_0$.
    By \Cref{rootsum} we conclude that $\gamma - \alpha_{i_0} \in \Phi \cup \{0\}$, and the case $\gamma - \alpha_{i_0} = 0$ (i.e., $\gamma = \alpha_{i_0}$) is impossible because $\gamma$ is assumed not already in~$A$.
\end{proof}

Finally, we need the following known fact~\cite{Hil16}:
\begin{fact} \label{fact:simplerootsum}
    Let $\Phi$ be an irreducible root system with simple roots $\Pi = \{\alpha_1, \dots, \alpha_d\}$.
    Then $\sum_{i=1}^d \alpha_i \in \Phi$.
\end{fact}

\subsection{Chevalley groups}
We may now define the Chevalley groups, via the \emph{Steinberg presentation} (see, e.g.,~\cite[Thm.~12.1.1]{Car89}).
\begin{definition}  \label{def:chevgroup}
    Corresponding to any irreducible root system $\Phi$ of rank at least~$2$, and any finite field $\F$, there is an associated \emph{universal (or simply connected) Chevalley group}, denoted~$\Esc{\Phi}{\F}$. Abstractly, it is generated by symbols $x_\alpha(t)$ for $\alpha \in \Phi$ and $t \in \F$, subject to the relations
    \begin{align*}
    x_\alpha(t)x_\alpha(u) &= x_\alpha(t+u)\\
    \phantom{\qquad \text{if } \alpha+\beta \neq 0 } %
    [x_\alpha(t),x_\beta(u)] &= \prod_{i,j>0} x_{i\alpha+j\beta}(C_{ij}^{\alpha,\beta}t^iu^j) \qquad \text{(for $\alpha+\beta \neq 0$)}\\
    h_\alpha(t)h_\alpha(u) &= h_\alpha(tu) \quad \text{(for $tu \neq 0$)},\\
    \text{where }\hspace{1cm} h_\alpha(t) &= n_\alpha(t)n_\alpha(-1)\\
    \text{and }\hspace{1cm} n_\alpha(t) &= x_\alpha(t)x_{-\alpha}(-t^{-1})x_\alpha(t).
    \end{align*}
    The second relation above is the \emph{Chevalley commutator formula}, and it is elaborated upon in \Cref{fact:chevcom} below.
\end{definition}
\begin{definition}  \label{def:adj}
    Let $\Center{\Phi}{\F}$ denote the center of $\Gu{\Phi}{\F}$.
    The \emph{adjoint Chevalley group}, which we denote by~$\G{\Phi}{\F}$, is the quotient $\Esc{\Phi}{\F}/\Center{\Phi}{\F}$.
    In all cases, $\Center{\Phi}{\F}$ is a constant-sized subgroup (of size $d+1$ for $\Phi = A_d$, and of size at most~$4$ otherwise).\footnote{Specifically, it is isomorphic to $\Z_{d+1}$ when $\Phi = A_d$, to $\Z_2$ when $\Phi \in \{B_d, C_d, E_7\}$, to $\Z_4$ or $\Z_2 \times \Z_2$ when $\Phi = D_d$ (for odd, even $d$ respectively), to $\Z_3$ when $\Phi = E_7$, and is trivial otherwise.~\cite[Sec.~3.3]{Ste16}.}
    It is generated by certain products~$\prod_{\alpha \in \Pi} h_{\alpha}(t_\alpha)$ (i.e., diagonal matrices in the matrix realizations), where $\Pi$ is a simple set of roots and the $t_\alpha$'s are roots of unity in~$\F$.
\end{definition}
\begin{remark}  \label{rem:CFSG}
   The Classification of Finite Simple Groups~\cite{Asc04} states that as $\F$ ranges over all finite fields, the adjoint Chevalley groups (excluding $\G{A_1}{\F_2}$, $\G{A_1}{\F_3}$, $\G{B_2}{\F_2}$, $\G{G_2}{\F_2}$, but including the ``twisted'' versions, which we do not discuss in this work) constitute the finite simple groups, together with the cyclic, alternating, and sporadic simple groups.
\end{remark}
Although, strictly speaking, it is the adjoint Chevalley groups that are the simple ones, it is more convenient to work with the very slightly larger universal Chevalley groups.
If one wants to precisely fulfill the goal concerning simple (adjoint) Chevalley groups described in \Cref{sec:goal}, one may use do so by appealing to \Cref{rem:get-to-adjoint}.
But henceforth we work exclusively with the universal Chevalley groups, and we will drop the adjective ``universal''.

Although we have defined the Chevalley groups abstractly, we have~\cite[Sec.~3.3]{Ste16} the isomorphisms with classical groups shown in \Cref{tab:iso}, for the ``classical" root systems of types $A$, $B$, $C$, and~$D$.

\begin{table}[H]
\centering
\begin{tabular}{cll}\label{table:classical}
Type of $\Phi$     & $\G{\Phi}{\cdot}$     & $\Esc{\Phi}{\cdot}$      \\ \hline
$A_d$      & $\mathrm{PSL}_{d+1}$  & $\mathrm{SL}_{d+1}$    \\
$B_d$      & $\mathrm{SO}_{2d+1}$  & $\mathrm{Spin}_{2d+1}$ \\
$C_d$      & $\mathrm{PSp}_{2d}$   & $\mathrm{Sp}_{2d}$     \\
$D_{2\ell}$   & $\mathrm{PSO}_{4\ell}$   & $\mathrm{Spin}_{4\ell}$\\
$D_{2\ell+1}$ & $\mathrm{PSO}_{4\ell+2}$ & $\mathrm{Spin}_{4\ell+2}$
\end{tabular}
\caption{The Chevalley groups corresponding to classical root systems.}
\label{tab:iso}
\end{table}

Identifications of the root elements $x_\alpha(t)$ of the Chevalley groups of classical type as elements of the corresponding matrix groups can be found in~\cite[Sec.~11.3]{Car89}, and matrix realizations for the exceptional Chevalley groups can be found in~\cite{HRT01}.

As we discuss in \Cref{sec:computation}, we have
\[
    n \coloneqq \abs{\Esc{\Phi}{\F}} = \abs{\F}^{\Theta(1)} = \exp(\Theta(m))
\]
for fixed~$p$ and~$\Phi$; indeed, an exact formula for $\abs{\Esc{\Phi}{\F}}$ is known, and one can compute within $\Esc{\Phi}{\F}$ (and $\Gad{\Phi}{\F}$) in $\poly(m) = \polylog(n)$ time (see \Cref{sec:computation}).

\begin{remark}
    From the first relation of \Cref{def:chevgroup}, it follows that the subgroup $\la x_\alpha(r) : r \in \F\ra$ of $\Esc{\Phi}{\F}$ is isomorphic to the additive group of $\F$. This subgroup is called the \emph{root subgroup} associated to~$\alpha$.
\end{remark}

The second relation in \Cref{def:chevgroup} will be used to give explicit descriptions of the links in our constructions, so we elaborate on it here.
\begin{theorem}\label{fact:chevcom}
    The \emph{Chevalley commutator formula} asserts that within $\Esc{\Phi}{\F}$ and $\G{\Phi}{\F}$, if $\alpha, \beta \in \Phi$ with $\alpha + \beta \neq 0$, and $t,u \in \F$, then 
    \[
        [x_\alpha(t), x_\beta(u)] = \prod_{\substack{i,j \in \Z_+ \\ i\alpha + j\beta \in \Phi}} x_{i\alpha + j\beta}(\Const{\alpha}{\beta}{i}{j} t^i u^j)
    \]
    for certain \emph{structure constants} $\Const{\alpha}{\beta}{i}{j} \in \{\pm 1, \pm 2, \pm 3\}$ that can be found in, e.g.,~\cite[Sec.~5.2]{Car89}. Here the product above is taken in order of increasing $i+j$.\footnote{Ties may be broken arbitrarily, as it turns out that  elements with equal $i+j$ commute.} 
    In addition, the structure constants only depend on  the set $\{(i,j):i\alpha+j\beta \in \Phi \}$.\rnote{I don't really know what this comment means}
\end{theorem}
\begin{remark}
In particular, the commutator formula implies that if $\alpha + \beta \not \in \Phi \cup \{0\}$, then $[x_\alpha(r), x_\beta(r)] = 1$.
\end{remark}
\begin{remark}\label{rmk:signs}
The constants $\Const{\alpha}{\beta}{i}{j}$ are determined uniquely by $\Phi$ up to signs. Different signs can arise from different choices of a \emph{Chevalley basis}. The resulting groups are isomorphic, however. See~\cite[Prop.~4.2.2]{Car89} and the preceding discussion.\rnote{I don't really know what this remark means.}
\end{remark}
Although  not strictly necessary for our work, we give explicit structure constants in the following description of the commutator formula for root systems of rank~$2$:
\begin{proposition}\label{rk2com}
    \textnormal{(\cite[Sec.~33.3--33.5]{Hum95}.)}
    Let $\Phi$ be one of $A_2,B_2,$ or $G_2$ and let $t,u \in \F$. Then:\footnote{For $G_2$ we fix the signs implemented in the GAP~\cite{GAP4} package Unipot~\cite{HH18}, which we used for calculations in \Cref{g2graphs}.}
    \begin{itemize}
        \item If $\Phi = A_2$ with positive roots $\alpha,\beta, \alpha+\beta$, then
        \[[x_\alpha(t),x_\beta(u)] = x_{\alpha + \beta}(tu).\]
        \item If $\Phi = B_2$ with positive roots $\alpha,\beta,\alpha+\beta,2\alpha+\beta$, then
        \begin{align*}
        [x_\beta(t),x_\alpha(u)] &= x_{\alpha + \beta}(tu)x_{2\alpha+\beta}(t^2u)\\
        [x_{\alpha+\beta}(t),x_\alpha(u)] &= x_{2\alpha+\beta}(2tu).
        \end{align*}
        \item If $\Phi = G_2$ with positive roots $\alpha,\beta,\alpha+\beta,2\alpha+\beta,3\alpha+\beta,3\alpha+2\beta$, then
        \begin{align*}
        [x_\beta(t),x_\alpha(u)] &= x_{\alpha+\beta}(tu)x_{2\alpha+\beta}(tu^2)x_{3\alpha+\beta}(tu^3)x_{3\alpha+2\beta}(-t^2u^3) \\
        [x_{\alpha+\beta}(t),x_{\alpha}(u)] &= x_{2\alpha+\beta}(2tu)x_{3\alpha+\beta}(3tu^2)x_{3\alpha+2\beta}(-3t^2u)\\
        [x_{2\alpha+\beta}(t),x_{\alpha}(u)] &= x_{3\alpha+\beta}(3tu)\\
        [x_{3\alpha+\beta}(t),x_{\beta}(u)] &= x_{3\alpha+2\beta}(tu) \\
        [x_{2\alpha+\beta}(t),x_{\alpha+\beta}(u)] &= x_{3\alpha+2\beta}(-3tu).
        \end{align*}
        \end{itemize}
\end{proposition}





Finally, we will require two more key facts:
\begin{proposition} \label{prop:steinberg17}
    \textnormal{(\cite[Lem.~17]{Ste16}.}
    In the Chevalley group $\Esc{\Phi}{\F}$, suppose $S \subset \Phi$ is a set of roots with the following two properties: (i)~$\alpha, \beta \in S$ and $\alpha + \beta \in \Phi$ implies $\alpha + \beta \in S$; (ii)~$\alpha \in S$ implies $-\alpha \not \in S$.
    Then each element of the subgroup $\la x_\alpha(t) : \alpha \in S, t \in \F \ra$ can be expressed uniquely as $\prod_{\alpha \in S} x_\alpha(t_\alpha)$ for some $t_\alpha \in \F$, where the product is taken in some fixed order (and this is true for any fixed ordering of~$S$ for the product).
\end{proposition}
\begin{proposition}\label{prop:uniint}
Let $\Pi$ be a set of simple roots, and define the two subgroups $U^\pm = \langle x_{\alpha}(t) : \alpha \in \Pi^\pm \rangle$. Then $U^+ \cap U^- = \{1\}$.
\end{proposition}
\begin{proof}
By \cite[Lem.~18, Cor.~3]{Ste16}, $\Esc{\Phi}{\F}$ can be realized as a group of matrices over $\F$ where the subgroup $U^+$ is upper-unitriangular and $U^{-}$ is lower-unitriangular. The proposition follows.
\end{proof}

\subsection{Computation within the Chevalley groups}    \label{sec:computation}
Given field $\F = \F_q = \F_{p^m}$ and root system $\Phi$ of rank~$d$, let us treat $\Phi$ and $p$ as fixed, and $m \to \infty$ as an asymptotically growing parameter.
Here we recap the known facts that the Chevalley group $\Gu{\Phi}{\F}$ has order $n = \exp(\Theta(m))$ and that one can compute within~$\ol{G}$ in deterministic $\poly(m) = \polylog(n)$ time.
(The same is true for the adjoint Chevalley group $\Gad{\Phi}{\F}$.)\\

First, field arithmetic is efficient, thanks to Shoup:
\begin{theorem}                                     \label{thm:fields}
    \textnormal{(\cite{Sho90}.)}
    For a fixed prime~$p$, there is an deterministic $\poly(m)$-time algorithm for finding an irreducible $f \in \F_p[x]$ of degree~$d$, and thereby ``constructing'' the field $\F = \F_q = \F_{p^m}$.
    The elements of~$\F$ are encoded by bit-strings of length~$\Theta(m)$, and field operations may be computed in deterministic $\poly(m)$ time --- this includes computing all $k$th roots of unity in $\poly(k,m)$ time.
\end{theorem}

Next, we note that there is an easy-to-compute formula for the order of a given Chevalley group:
\begin{theorem}                                     \label{thm:order}
    For $\Phi$ of rank~$d$, the order of the group $\Gu{\Phi}{\F}$ is of the form $q^{\Theta(d^2)} = p^{\Theta(d^2 m)}$, where the constant hidden in the $\Theta(\cdot)$ depends only on~$\Phi$.
    Moreover there is a precise formula for $\abs{\Gu{\Phi}{\F}}$ that can easily be computed in $\poly(d, \log p, m)$ time; see, e.g.~\cite[Thm.~25]{Ste16}.
    (All of this is also true of  $\Gad{\Phi}{\F}$.)
\end{theorem}

Finally, we appeal to the work of Cohen, Murray, and Taylor~\cite{CMT04} to show that one can efficiently construct and compute within Chevalley groups:
\begin{theorem}                                     \label{thm:computing}
    \textnormal{(\cite{CMT04}, see especially Sec.~8.1.)}
     For $\Phi$ of rank~$d$ and $\F = \F_{p^m}$, there is a canonical representation (``Bruhat normal form'') for each element of $\Gu{\Phi}{\F}$, encoded by a bit-string of length $\poly(d, \log q, m)$.  One can pass between this form, a natural matrix representation, and an expression in the Steinberg presentation --- and also compute group products and inverses --- via deterministic $\poly(d, \log q, m)$-time algorithms.
     (Since $k$th roots of unity can also be computed efficiently (\Cref{thm:fields}), the $O(d)$-size center~$Z$ of $\Gu{\Phi}{\F}$ can also be constructed efficiently, and hence this whole theorem is also true for~$\Gad{\Phi}{\F}$.)
\end{theorem}

\section{The Construction}\label{sec:construct}
For the rest of the paper we fix a field $\F$ of size $p^m$ where $p > 3$, an irreducible root system $\Phi$ of rank at least $2$, and a set of simple roots $\Pi = \{\alpha_1, \dots, \alpha_d\} \subset \Phi$. With this in mind, $x_\alpha(t)$ refers to the corresponding root element of $\Esc{\Phi}{\F}$.

\begin{definition}
    For $S \subseteq \Phi$ and $d \in \N$, let $X_{S,d} = \langle x_\alpha(t) : \alpha \in S, t \in \F, \deg(t) \le d \rangle$. For shorthands we write $X_S = X_{S,1}$ and also $X_{\alpha,d} = X_{\{\alpha\},d}$.
\end{definition}

\begin{definition}\label{def:rootchoice}
    Recalling $\Pi = \{\alpha_1, \dots, \alpha_d\}$, we define  $\calS$ to be the following particular set of roots:
     \begin{equation}   \label{eqn:ourS}
        \calS = \Pi \cup \left\{-(\alpha_1 + \cdots + \alpha_d)\right\}.
     \end{equation}
     (The last of these is a root by by \Cref{fact:simplerootsum}.)
\end{definition}

\begin{remark}  \label{rem:matr}
    Since $\Pi$ is a basis, it follows that every subset of $\calS$ of cardinality~$d$ is linearly independent.
\end{remark}

\begin{definition}
    For each $\alpha \in \calS$, we introduce the following subgroup of $\Gu{\Phi}{\F}$:
    \[
        H_\alpha= X_{\calS \setminus \{ \alpha\}}.
    \]
\end{definition}

Finally, we can introduce our coset complex:
\begin{definition}
    $\displaystyle \complex = \complex_m := \CC{\Gu{\Phi}{\F}}{(H_\alpha)_{\alpha \in \calS}}.$
\end{definition}

\begin{theorem}\label{thm:main}
    For $d = \rank(\Phi)$, it holds that $\complex$ is a $d$-dimensional pure simplicial complex, where:
\begin{enumerate}
\item $\abs{\complex(0)} = p^{\Theta(m)}$, where the constant hidden by $\Theta(\cdot)$ depends only on~$\Phi$; moreover, the family that arises as $m \to \infty$ is strongly explicit.
\rnote{would be good to quantify.  For example, I believe that when you increment $m$ by $+1$, the size of $\complex(0)$ grows by $p^{\poly(d)}$, where the polynomial depends on~$\Phi$.}
\item Every vertex participates in at most $\Delta = \Delta(\Phi, p) = p^{\Theta(1)}$ maximal faces,  where the $\Theta(\cdot)$ constant (independent of~$m$) depends only on~$\Phi$ (indeed, it is $\Theta(d^2)$).
\item If $p > 3$,\footnote{Recall \Cref{foot:5}.} then $\Link_\sigma(\complex)$ is connected for all $\sigma \in \complex(i)$, $i \leq d-2$.
\item If $\Phi \neq G_2$ and $p > 2$, then for all $\sigma \in \complex(d-2)$ it holds that $K_\sigma$ is a $p^2$-regular bipartite graph with $\lambda_2(K_\sigma) \leq \sqrt{2/p}$.
\item $\Esc{\Phi}{\F}$ acts simply transitively on the maximal faces of $\complex$ (and this is also true if one constructs $\complex$ from $\Gad{\Phi}{\F}$ rather than $\Esc{\Phi}{\F}$).
\end{enumerate}
\end{theorem}

By \Cref{thm:trickle}, we conclude our final goal:
\begin{corollary}
    Fixing $\Phi \neq G_2$ of rank $d \geq 2$, $p>3$ prime, and taking $m \to \infty$, the sequence $(\complex_m)$ forms a strongly explicit $d$-dimensional, $\Delta$-bounded degree ($\Delta = p^{\Theta(d^2)}$), $\lambda$-spectral HDX family, where
	\[\lambda \le \frac{1}{\sqrt{p/2} - d + 1}.\]
    (Hence for large~$p$, we have $\lambda \sim 1/\Delta^{\Theta(1/d^2)}$.)
    Moreover, the universal Chevalley group $\Esc{\Phi}{\F}$ acts simply transitively on $\complex_m$'s maximal faces.
\end{corollary}
We add that the results all remain true if uses the (simple) adjoint Chevalley groups $\Gad{\Phi}{\F}$ in place of~$\Esc{\Phi}{\F}$.

\subsection{Global connectivity of the coset complex}

The main goal of this section is to show that the subgroups $H_\alpha$ for $\alpha \in \calS$ generate $\Esc{\Phi}{\F}$. By \Cref{fact:connected}, this is necessary to ensure that the 1-skeleton of $\complex$ is connected.

\begin{theorem}\label{thm:rootgen}
    Let $S \subseteq \Phi$ be a subset of $\rank(\Phi)+1$ roots where $S^+ = \Phi$. Then $X_S = \Gu{\Phi}{\F}$.
\end{theorem}

The particular set of roots $\calS$ we selected in \Cref{eqn:ourS} has the desired property, as the following shows:
\begin{proposition}\label{prop:sspan}
    For $\calS$ as in \textnormal{\Cref{eqn:ourS}} we have $\calS^+ = \Phi$.
\end{proposition}
\begin{proof}
    We  have $\calS^+ \supseteq \Pi^+$, and so $\calS^+$ certainly contains all positive roots in~$\Phi$ (recall \Cref{def:simple}).
    It remains to show that $\calS^+$ contains each negative root~$\gamma \in \Phi$.
    Writing $\gamma = -n_1 \alpha_1 - \cdots -n_d \alpha_d$, it follows that we can reexpress it as
    \[
        \gamma = r (-(\alpha_1 + \cdots + \alpha_d)) + r_1 \alpha_1 + \cdots + r_d \alpha_d
    \]
    for a sufficiently large positive integer~$r$, and positive integers $r_1, \dots, r_d$.
    Thus indeed $\gamma \in \calS^+$.
\end{proof}
\begin{example}
    $A_d$ is the set of vectors $\{e_i -e_j, i \neq j\} \subset \R^d$. A set of simple roots is given by $ \Pi = \{e_i-e_{i+1} : i \in [d]\}$; in this case $-\sum_{\alpha \in \Pi} \alpha = e_d-e_1$. It is straightforward to check that $S = \{e_i-e_{i+1} : i \in [d]\}\cup \{e_d-e_1\} \subset A_d$ satisfies the hypothesis of \Cref{thm:rootgen}. This is the set of roots implicitly used in~\cite{KO18}.
\end{example}
\begin{remark}  \label{rem:alt}
    There are other choices of $S$ besides our $\calS$ from \Cref{eqn:ourS} that satisfy the condition of \Cref{thm:rootgen}.
    These can be used to obtain slightly different constructions.
    For example, referring to \Cref{fig:root-eg} one see that in~$B_2$ one can take $S = \{\alpha, \beta, -\beta-2\alpha\}$, or in~$G_2$ one can take $S = \{\alpha, \alpha+\beta, -2\alpha - \beta\}$.
\end{remark}

We will require the following (presumably known) fact:
\begin{lemma}\label{lem:powerspan}
    For $i,j,d_1,d_2 \in \N$ with $\mathrm{char}(\F) > \max(i,j)$, write $d = i d_1 + j d_2$. Then 
    \[
        \F[x]^{\leq d} = \mathrm{span}\{ f^i g^j : f \in \F[x]^{\leq d_1},\ g \in \F[x]^{\leq d_2}\}.
    \]
    where $\F[x]^{\leq k}$ represents the polynomials of degree at most~$k$.
\end{lemma}
\begin{proof}
It suffices to establish that $x^e$ is in the span, for any $e \leq d$.
Express $e = a_1 + \cdots + a_i + b_1 + \cdots + b_j$, with each $a$ being a natural number at most~$d_1$ and each $b$ being a natural number at most~$d_2$.
Now note that the monomial
\begin{equation} \label{eqn:aaa}
    x_{a_1} x_{a_2} \cdots x_{a_i} x_{b_1} x_{b_2} \cdots x_{b_j}
\end{equation}
becomes equal to $x^e$ if each indeterminate $x_c$ is substituted with~$x^c$.
Next, we use the identity
\[
    x_{a_1} x_2 \cdots x_{a_i} = \frac{1}{i!} \sum_{s \in \{0,1\}^i } (-1)^{|s|+i} \left (\sum_{\ell=1}^i s_\ell x_{a_\ell} \right)^i,
\]
with the constant $\frac{1}{i!}$ being sensible in the field~$\F$ since $\mathrm{char}(\F) > i$.
(This is the ``higher order polarization identity'', or Ryser's formula applied to the matrix where every row is $\begin{bmatrix} x_{a_1} & x_{a_2} & \cdots & x_{a_i} \end{bmatrix}$.)
Multiplying this against the analogous identity with the $b$'s (and using $\mathrm{char}(\F) > j$), we get that~\eqref{eqn:aaa} can be expressed as a linear combination of multivariate polynomials $F^i G^j$, where each~$F$ is a linear combination of $x_c$'s with $c \leq d_1$ and each~$G$ is a linear combination of $x_c$'s with $c \leq d_2$.
Now substituting $x_c = x^c$ yields the desired univariate expression for~$x^e$.
%
%
%
\end{proof}

A key goal now is to establish the below \Cref{rk2gen}.  We remark that several times it will use \Cref{lem:powerspan}; in each application we will have ``$i$'' and ``$j$'' at most~$3$, less than $\mathrm{char}(\F) = p > 3$ as required.
\begin{lemma}\label{rk2gen}
    Fix roots $\beta \neq -\alpha$  and any $d_1,d_2 \in \N$. Then
    \[
        \langle X_{\alpha,d_1},X_{\beta,d_2} \rangle= \langle X_{i\alpha + j\beta,id_1+jd_2}  : i,j \in \N, i\alpha+j\beta \in \Phi \rangle.  
    \]
\end{lemma}
\begin{proof}
    The inclusion $\subseteq$ is immediate by taking $(i,j) \in \{(1,0), (0,1)\}$, so it suffices to prove the reverse inclusion~$\supseteq$.
    The case $\beta = \alpha$ is trivial, so we may assume that $\alpha, \beta$ span some $2$-dimensional subspace~$H$.
    Let $R \coloneqq \{i\alpha + j\beta \in \Phi : i,j \in \N\}$, a subset of the $2$-dimensional root system $\Phi' = \Phi \cap H$.
    If $R = \{\alpha,\beta\}$ only then the lemma is immediate. 
    Otherwise, $R$ must also contain $\alpha + \beta$ (using \Cref{fact:prefixsum}) and hence $\Phi'$ is isomorphic to $A_2$, $B_2$, or $G_2$ as explained in \Cref{rem:2rootspan}.
    This allows us to classify the possibilities for~$R$; 
    with the assistance of \Cref{fig:root-eg}, we see there are four cases, namely $R = \{\alpha,\beta\} \cup R'$ for $R'$ equal to\dots
    \[
        1.\ \{\alpha+\beta\}, \quad 
        2.\ \{\alpha+\beta, 2\alpha+\beta\}, \quad
        3.\ \{\alpha+\beta, 2\alpha+\beta,\alpha+2\beta\}, \quad
        4.\ \{\alpha+\beta, 2\alpha+\beta, 3\alpha+\beta, 3\alpha +2\beta\}.
    \]
    In each case, we need to show for every $\gamma =i\alpha + j\beta \in R'$ that $x_{\gamma}(w) \in \langle X_{\alpha,d_1}, X_{\beta,d_2}\rangle$ for all $w \in \F$ of degree at most $d = id_1 + jd_2$.  By virtue of \Cref{lem:powerspan} (and using $\mathrm{char}(\F) > 3 \geq i,j$), it suffices to show this for~$w$'s that are linear combinations of field elements of the form~$t^iu^j$, where $t$~has degree~$d_1$ and $u$ has degree~$d_2$.  Further, since $x_\gamma(r+s) = x_\gamma(r)x_\gamma(s)$, it suffices to handle $w$~of the form $c t^i u^j$ for arbitrary $c \in \F_p$.  Finally, it suffices to handle just one specific $c \neq 0$, because if $x_\gamma(c t^i u^j)$ is in $\langle X_{\alpha,d_1}, X_{\beta,d_2}\rangle$ then so too is its~$k$th power $x_\gamma(c t^i u^j)^k = x_\gamma(kc t^i u^j)$, and $kc$ varies over all $\F_p$ as $k$ varies in~$\N$.  We will always use a~$c$ which is the product of structure constants $C^{\alpha',\beta'}_{i',j'}$, and such are never~$0$ in $\F_p$ because $1 \leq |C^{\alpha',\beta'}_{i',j'}| \leq 3  < p$.
    
    Summarizing, for fixed $t, u \in \F$ of degree at most $d_1, d_2$ (respectively), it suffices to show the following in Cases~1--4: For each  $\gamma = i\alpha + j\beta \in R'$ we have $x_\gamma(c t^iu^j) \in \langle x_{\alpha}(t), x_{\beta}(u) \rangle$ for some product of structure constants~$c$.
    
    \paragraph{Case 1:} $R' = \{\alpha+\beta\}$. This case arises when $\Phi' = A_2$ and $\angle(\alpha,\beta) = 120^\circ$, or when $\Phi' = B_2$ and $\alpha,\beta$ are short roots with $\angle(\alpha,\beta) = 90^\circ$, or when $\Phi' = G_2$ and $\alpha,\beta$ are short roots with $\angle(\alpha,\beta) = 60^\circ$.  We handle $\gamma = \alpha + \beta$ via the commutator formula
    $
        [x_\alpha(t),x_\beta(u)] = x_{\alpha+\beta}(C_{1,1}^{\alpha,\beta}tu).
    $
    
    \paragraph{Case 2:} $R' = \{ \alpha+\beta, 2\alpha + \beta\}$, which arises for $\Phi' = B_2$.  We first treat the root $\gamma = 2\alpha+\beta$.  By the commutator formula we have
    \[
        [[x_\alpha(t),x_\beta(u)],x_\alpha(t)] = [x_{\alpha+\beta}(C_{1,1}^{\alpha,\beta}tu) x_{2\alpha+\beta}(C_{2,1}^{\alpha,\beta}t^2u), x_\alpha(t)].
    \]
    In  this latter  commutator we can delete~$x_{2\alpha+\beta}(C_{2,1}^{\alpha,\beta}t^2u)$ because it commutes with the other two elements. (This is since no root is a nontrivial $\N$-linear combination involving $2\alpha+\beta$.)  Thus
    \begin{equation}    \label{C2a}
        [[x_\alpha(t),x_\beta(u)],x_\alpha(t)] = 
        [x_{\alpha+\beta}(C_{1,1}^{\alpha,\beta}tu), x_\alpha(t)] = 
        x_{2\alpha+\beta}(C_{1,1}^{\alpha+\beta,\alpha}C_{1,1}^{\alpha,\beta}t^2u).
    \end{equation}
    Thus $\gamma = 2\alpha+\beta$ is handled.  As for $\gamma = \alpha + \beta$, the commutator formula gives 
    \[
        [x_\alpha(t),x_\beta(u)]\cdot x_{2\alpha+\beta}(-C_{2,1}^{\alpha,\beta}t^2u) =  x_{\alpha+\beta}(C_{1,1}^{\alpha,\beta}tu)x_{2\alpha+\beta}(C_{2,1}^{\alpha,\beta}t^2u)\cdot x_{2\alpha+\beta}(-C_{2,1}^{\alpha,\beta}t^2u)
        = x_{\alpha+\beta}(C_{1,1}^{\alpha,\beta}tu),
    \]
    and so $\gamma = \alpha+\beta$ is also handled (since we already know $x_{2\alpha+\beta}(-C_{2,1}^{\alpha,\beta}t^2u)$ is in $\langle x_\alpha(t), x_\beta(u) \rangle$ via \Cref{C2a}).

    \paragraph{Case 3:} $R' = \{\alpha+\beta, 2\alpha+\beta, \alpha+2\beta\}$. This case only arises for $\Phi' = G_2$. We start by treating $\gamma = 2\alpha+\beta$.  We have 
    \[
        [[x_\alpha(t),x_\beta(u)],x_\alpha(t)] = [x_{\alpha+\beta}(C_{1,1}^{\alpha,\beta}tu) y, x_\alpha(t)] \quad \text{for } y = x_{2\alpha+\beta}(C_{2,1}^{\alpha,\beta}t^2u)x_{\alpha+2\beta}(C_{3,1}^{\alpha,\beta}tu^2),
    \]
    and similar to Case~2 we can delete $y$ from this commutator as it commutes with the other two elements (by virtue of the height of $2\alpha+\beta$ and $\alpha + 2\beta$).  Hence
    \[
        [[x_\alpha(t),x_\beta(u)],x_\alpha(t)] =  [x_{\alpha+\beta}(C_{1,1}^{\alpha,\beta}tu), x_\alpha(t)] = x_{2\alpha+\beta}(C_{1,1}^{\alpha,\beta}C_{1,1}^{\alpha+\beta,\alpha}t^2u)
    \]
    and we've handled $\gamma = 2\alpha + \beta$. The case of $\gamma = \alpha = 2\beta$ is similar.  Finally the treatment of  $\gamma = \alpha + \beta$ is similar to Case~2; it follows from
    \[[x_\alpha(t),x_\beta(u)]x_{\alpha+2\beta}(-C_{3,1}^{\alpha,\beta}tu^2)x_{2\alpha+\beta}(C_{2,1}^{\alpha,\beta}t^2u) = x_{\alpha+\beta}(C_{1,1}^{\alpha,\beta}tu).\]

    \paragraph{Case 4:} $R' = \{\alpha+\beta, 2\alpha+\beta,3\alpha+\beta, 3\alpha+2\beta\}$. This case only arises for $\Phi' = G_2$.   To reduce clutter in this case, we will sometimes abbreviate $x_{i\alpha+j\beta}(c t^i u^j)$ to $x_{i\alpha+j\beta}$.
    We start with
    \begin{equation}
        [x_\alpha(t),x_\beta(u)] = x_{\alpha+\beta} \cdot x_{2\alpha+\beta} \cdot x_{3\alpha+\beta} \cdot x_{3\alpha+2\beta},  \label{eqn:ugh0} 
    \end{equation}
    which  implies
    \[
        [[x_\alpha(t), x_\beta(u)], x_\beta(u)] = 
        [ x_{\alpha+\beta}\cdot x_{2\alpha+\beta}\cdot x_{3\alpha+\beta}, x_\beta ], 
    \]
    where we deleted the $x_{3\alpha+2\beta}$ element since it commutes with everything else.  Now since $x_\beta$ commutes with $x_{\alpha+\beta}$ and $x_{2\alpha+\beta}$, we get 
    \[
        [ x_{\alpha+\beta}\cdot x_{2\alpha+\beta}\cdot x_{3\alpha+\beta}, x_\beta ] = [x_{3\alpha+\beta},x_\beta] = x_{3\alpha+2\beta} = x_{3\alpha+2\beta}(C_{1,1}^{3\alpha+\beta,\beta} C_{3,1}^{\alpha,\beta} t^3 u^2),
    \]
    where in the last step we explicitly wrote in the argument to $x_{3\alpha+2\beta}$ that arises.  Thus we have  handled $\gamma = 3\alpha + 2\beta$.  Taking care of $\gamma = 3\alpha + \beta$ is somewhat more tedious.  Considerations similar to the above lead us to 
    \[
        [[x_\alpha(t), x_\beta(u)], x_\alpha(t)] = 
        [ x_{\alpha+\beta}\cdot x_{2\alpha+\beta}, x_\alpha ],
    \]
    which in turn equals
    \begin{equation}    \label{eqn:ugh1}
        x_{2\alpha+\beta}(-C_{2,1}^{\alpha,\beta} t^2 u) \cdot [x_{\alpha+\beta},x_\alpha]\cdot x_{2\alpha+\beta}(C_{2,1}^{\alpha,\beta} t^2 u) \cdot [x_{2\alpha+\beta},x_\alpha],
    \end{equation}
    where we explicitly wrote in the arguments to $x_{2\alpha+\beta}$ that arise.  We now observe that when the commutator rule is twice applied in the above, the resulting elements are $x_{2\alpha+\beta} \cdot x_{3\alpha+2\beta} \cdot x_{3\alpha+\beta}$ (first commutator) and $x_{3\alpha+\beta}$ (second commutator), and these all commute with the $x_{2\alpha+\beta}(\pm C_{2,1}^{\alpha,\beta} t^2 u)$ in \Cref{eqn:ugh1}.  Thus said $x_{2\alpha+\beta}(\pm C_{2,1}^{\alpha,\beta} t^2 u)$ cancel out, and we end up deducing that 
    \begin{multline}    \label{eqn:ugh2}
         [[x_\alpha(t), x_\beta(u)], x_\alpha(t)] \\ =    
         x_{2\alpha+\beta}(C_{1,1}^{\alpha+\beta,\alpha}C_{1,1}^{\alpha,\beta}t^2u)
         x_{3\alpha+2\beta}(C_{2,1}^{\alpha+\beta,\alpha}C_{1,1}^{\alpha,\beta} t^3u^2)
         x_{3\alpha+\beta}((C_{1,2}^{\alpha+\beta,\alpha}C_{1,1}^{\alpha,\beta} +C_{1,1}^{2\alpha+\beta,\alpha}C_{2,1}^{\alpha,\beta})t^3u).
    \end{multline}
    Finally, we take one more commutator with $x_\alpha(t)$. The latter two elements in the above commute with $x_\alpha(t)$ and thus may be deleted; we are left with
    \[
        [[[x_\alpha(t), x_\beta(u)], x_\alpha(t)], x_\alpha(t)] =[x_{2\alpha+\beta}(C_{1,1}^{\alpha+\beta,\alpha}C_{1,1}^{\alpha,\beta}t^2u), x_{\alpha(t)}] = x_{3\alpha+\beta}(C_{1,1}^{2\alpha+\beta,\alpha}C_{1,1}^{\alpha+\beta,\alpha}C_{1,1}^{\alpha,\beta}t^3u).
    \]
    Thus we have handled $\gamma = 3\alpha+\beta$.  Since $\gamma = 3\alpha+2\beta$ has also been treated, we get $\gamma = 2\alpha+\beta$ from \Cref{eqn:ugh2}, and then $\gamma = \alpha+\beta$ from \Cref{eqn:ugh0}.
\end{proof}

We may now complete our goal for this section:
\begin{proof}[Proof of \Cref{thm:rootgen}]
We first show that $X_{\alpha} = X_{\alpha,1} \subseteq X_S$ for all $\alpha \in \Phi$.
Since we are assuming $S^+ = \Phi$, we can write $\alpha = \sum_{\beta \in S} n_\beta \beta$ with $n_\beta \in \N$.
Then by \Cref{fact:prefixsum} we can write $\alpha = p_{i_1}+p_{i_2}  + \cdots + p_{i_\ell}$ with $p_{i_j} \in S$ and where all prefix sums are roots.
Clearly we may assume that $p_{j} \neq -(p_1 + \cdots + p_{j-1})$ does not occur for any~$j$, as otherwise the first $j$ terms could be excised from the expression for~$\alpha$.
Then by \Cref{rk2gen} it follows that $X_{p_{i_1} + p_{i_2}} \subseteq \langle X_{p_{i_1}},X_{p_{i_2}} \rangle$, $X_{p_{i_1}+p_{i_2}+p_{i_3}} \subseteq \langle X_{p_{i_1}},X_{p_{i_2}},X_{p_{i_3}} \rangle$, and so on, eventually yielding $X_\alpha \subseteq \langle X_\beta : \beta \in \Phi \rangle$.

Now suppose by induction on $i \geq 0$ that $X_{\alpha,2^i} \subseteq X_S$ for all $\alpha \in \Phi$. By \Cref{fact:rootdecomp}, for any  root $\gamma \in \Phi$ we can write $\gamma = \alpha + \beta$ for some $\alpha, \beta \in \Phi$, and it follows from \Cref{rk2gen} that $X_{\gamma, 2^i+2^i} \subseteq \langle X_{\alpha,2^i},X_{\beta,2^i} \rangle$. Thus indeed $X_{\gamma,2^{i+1}} \subseteq X_S$, completing the induction.
\end{proof}

\subsection{Structure of the links}
In this section we describe the structure of the subgroups $X_T$ where $T \subset \mathcal{S}$. This will be used to show that the links of all faces of $\complex$ are connected.

We will first need a ``graded" version of \Cref{prop:steinberg17}.
\begin{proposition}\label{substruct}
    Fix any ordering~$\prec$ of the roots $\Phi$, and let $\Psi \subseteq \Phi$ be linearly independent.  Then the elements of $X_{\Psi}$ are in $1$-$1$ correspondence with expressions of the form $\prod_{\gamma \in \Psi^+} x_{\gamma}(t_\gamma)$ with  $x_\gamma(t_\gamma) \in X_{\gamma,\hgt_\Psi(\gamma)}$ (and the product taken in order~$\prec$).
\end{proposition}
\begin{proof}
    We first prove that every expression of the given form is indeed in~$X_\Psi$.  Precisely, we show by induction on~$h$ that $X_\Psi$ contains all subgroups $X_{\gamma,h}$ with $h = \hgt(\gamma)$.
    The base case of $h = 1$ is immediate.
    For general~$h$, take any $\gamma \in \Psi^+$ with height~$h$ and write $\gamma = \alpha + \beta$ with $\alpha,\beta \in \Psi^+$ of height smaller than~$h$.  (This is possible by \Cref{fact:prefixsum}.)  Now it follows from \Cref{rk2gen} that $X_{\gamma,\hgt(\gamma)} = X_{\gamma,\hgt(\alpha)+\hgt(\beta)} \subseteq \langle X_{\alpha, \hgt(\alpha)}, X_{\beta, \hgt(\beta)} \rangle$, and this is in $X_\Psi$ by induction.  

    We next show that every element in~$X_\Psi$ has a unique expression of the given form.
    In fact, it suffices to show existence, since uniqueness follows from \Cref{prop:steinberg17} (note that $\Psi^+$ satisfies its hypotheses). 
    Let us say that an expression of the form
    \begin{equation}    \label{eqn:wellb}
        x_{\gamma_1}(u_1) x_{\gamma_2}(u_2)\cdots x_{\gamma_m}(u_m)
    \end{equation}
    with $\gamma_i \in \Psi^+$ is \emph{well-bounded} if each $u_i$ has degree at most $\hgt(\gamma_i)$.
    The desired existence result is that every $z \in X_{\Psi}$ has a well-bounded expression as above, where $\gamma_1, \dots, \gamma_m$ list the elements of~$\Psi^+$ in the order~$\prec$.  (We remark that it doesn't matter whether we are allowing consecutive duplicate $\gamma_i$'s in this list, since $x_{\gamma}(u) x_{\gamma}(u') = x_{\gamma}(u+u')$ and this preserves well-boundedness.)

    To show this existence, it actually suffices to repeat the existence proof in \Cref{prop:steinberg17}.  At a high level, this works because that proof ultimately only uses the commutator formula, and applications of the commutator formula preserve well-boundedness.  That is, starting from an arbitrary $z \in X_\Psi$, by definition we may express $z$ as in \Cref{eqn:wellb} with each $\gamma_i \in \Psi$ and each $u_i$ of degree at most~$1$.  This is well-bounded.  Then  an application of the commutator formula switches some  consecutive $x_{\gamma}(u)x_{\gamma'}(u')$ to $x_{\gamma'}(u')x_{\gamma}(u)[x_{\gamma}(u),x_{\gamma'}(u')]$, and this commutator is the product of elements of the form $x_{i\gamma + j \gamma'}(C_{i,j}^{\gamma,\gamma'} u^i (u')^j)$.  But this product is indeed well-bounded, presuming the former expression was well-bounded.

    For completeness, we sketch why the existence result in \Cref{prop:steinberg17} only relies on the commutator formula.
    We prefer to first follow the existence result in \cite[Thm.~5.3.3]{Car89}, which assumes that the order~$\prec$ is consistent with heights (meaning $\hgt_{\Psi}(\alpha) \leq \hgt_{\Psi}(\beta)$ implies $\alpha \prec \beta$).  Under this assumption, we may repeatedly reorder consecutive products $x_{\gamma}(u)x_{\gamma'}(u')$ whenever $\gamma' \prec \gamma$, as described above.  Notice that the new products of elements of the form $x_{i\gamma + j \gamma'}(C_{i,j}^{\gamma,\gamma'} u^i (u')^j)$ that arise are have $\hgt(i\gamma + j \gamma') > \hgt(\gamma), \hgt(\gamma')$.  Because of this, and the height-respecting property of~$\prec$, this process must eventually terminate with a (well-bounded) expression like \Cref{eqn:wellb} where the roots~$\gamma_i$ are in the order~$\prec$ (and any missing root $\gamma \in \Phi^+$ can be inserted via $x_{\gamma}(0)$).  
    
    It remains to treat the case that the root order~$\prec$ does \emph{not} necessarily respect heights. For this we appeal to \cite[Lem.~18]{Ste16}, the associated component of the proof of \Cref{prop:steinberg17}.
    It says that it suffices to check --- when $\gamma$ \emph{is} a height-respecting order, and  $\Psi^+ = \{\gamma_1 \prec \gamma_2 \prec \cdots \prec \gamma_m\}$ --- that each subgroup of the form
    \[
        B_i \coloneqq X_{\gamma_i, \hgt(\gamma_i)} \cdot X_{\gamma_{i+1}, \hgt(\gamma_{i+1})} \cdots X_{\gamma_r, \hgt(\gamma_m)} 
    \]
    is normal in $X_\Psi$.  To see this, take a generic well-bounded expression
    \[
        y = x_{\gamma_i}(t_i) x_{\gamma_{i+1}}(t_{i+1})\cdots x_{\gamma_{m}}(t_m)
    \]
    in~$B_i$ and consider conjugating it by an arbitrary well-bounded expression~$w$ as in \Cref{eqn:wellb}.  We have $w^{-1} y w = y [y,w]$, and expanding the commutator yields a well-bounded expression consisting only of $x_{\gamma}(v)$'s where $\hgt(\gamma) \geq \hgt(\gamma_i)$.  Now as in the previous argument, this may be further rearranged into a well-bounded expression in~$B_i$, showing that $B_i$ is closed under conjugation and hence normal.
\end{proof}

We have the following immediate consequence:
\begin{corollary}\label{subsize}
    Let $\Psi$ be a set of linearly independent roots. Then
$\displaystyle |X_{\Psi}| = \prod_{\alpha \in \Psi^+} p^{\hgt_{\Psi}(\alpha) + 1}.$
\end{corollary}
Importantly, $|X_{\Psi}|$ can be bounded independently of $m$ (where recall $|\F|=p^m$). This will imply that a vertex in $\complex$ belongs to just $p^{O(1)}$ faces where the $O(1)$ does  not depend on $m$.\\

The proceeding normal form result also helps us show the following:
\begin{proposition}\label{subint}
Let $\Psi$ and $\Psi'$ be sets of linearly independent roots. Then $X_\Psi \cap X_{\Psi'} = X_{\Psi \cap \Psi'}$.
\end{proposition}
\begin{proof}
    By \Cref{fact:simplechoice} we may choose a set $\Pi$ of simple roots with $\Psi \subseteq \Pi^+$. We apply \Cref{substruct} to any $g \in X_\Psi$ and $h \in X_{\Psi'}$, writing them as $g = \prod_{\alpha \in \Psi^+} x_\alpha(t_\alpha)$ and $h = \prod_{\alpha \in \Psi'^+} x_\alpha(u_\alpha)  = U \cdot L$, where we have ordered $h$ as a product $U$ of root elements in $\Pi^+$ times a product $L$ of root elements in $\Pi^-$. Now supposing $g=h$, we get $U^{-1}g=L$. But by \Cref{prop:uniint}, the only way this equality can hold is if $L = 1$. Hence we have $\prod_{\alpha \in \Psi^+} x_\alpha(t_\alpha) = \prod_{\alpha \in \Psi'^+} x_\alpha(u_\alpha)$, where on both sides $\alpha$ is ranging in $\Pi^+$; hence by uniqueness of these expressions (assuming the products are taken in the same order), equality holds just when $t_\alpha = u_\alpha$ for all $\alpha$. So the elements of $X_{\Psi} \cap X_{\Psi'}$ are exactly the elements of the form
\[\prod_{\alpha \in \Psi^+ \cap \Psi'^+} x_\alpha(f_\alpha) \]
    where $\deg(f_\alpha) \le \min(\hgt_{\Psi}(\alpha),\hgt_{\Psi'}(\alpha))$. But note that $\Psi^+ \cap \Psi'^+=(\Psi \cap \Psi')^+$ and $\hgt_\Psi(\alpha) = \hgt_{\Psi'}(\alpha)=\hgt_{\Psi \cap \Psi'}(\alpha)$ for $\alpha \in (\Psi \cap \Psi')^+$ due to linear independence. So any such an element belongs to $X_{\Psi \cap \Psi'}$ (using \Cref{substruct} again), which proves the proposition.
\end{proof}
\begin{observation}\label{obs:centerint}
In fact, $Z \cdot X_{\Psi} \cap Z \cdot X_{\Psi'} = Z \cdot X_{\Psi \cap \Psi'}$, where $Z$ denotes the center of $\Esc{\Phi}{\F}$. The proof proceeds in the same fashion: Under the matrix identification of \Cref{prop:uniint}, $Z$~consists of diagonal matrices. Thus if $D_1 U^{-1}g = D_2 L$ with $D_1$ and $D_2$ diagonal, $L$ lower-unitriangular and $U^{-1}g$ upper-unitriangular, we must have $D_1=D_2$ and $L=U^{-1}g=1$. This implies $\prod_{\alpha \in \Psi^+} x_\alpha(t_\alpha) = \prod_{\alpha \in \Psi'^+} x_\alpha(u_\alpha)$, and the rest of the proof follows as before.
\end{observation}

Combining \Cref{subint} with \Cref{fact:link} lets us understand the structure of the links in~$\complex$:
\begin{theorem} \label{thm:linkstruct}
    Let $\sigma \in \complex$ be a face of type $T \subsetneq \calS$.
    Then the link of~$F$ is isomorphic to the coset complex $\CC{X_{\calS \setminus T}}{(X_{\calS \setminus T \setminus \{\alpha\}} : \alpha \in \calS \setminus T)}$.
\end{theorem}
\begin{proof}
    For $T = \emptyset$, this is the combination of \Cref{thm:rootgen} and \Cref{prop:sspan}.  Otherwise, by virtue of \Cref{fact:link} it suffices to show that for any $U \subseteq \calS$,
    \[
        H_U = \bigcap_{\alpha \in U} H_\alpha =  \bigcap_{\alpha \in U} X_{\calS \setminus \{\alpha\}} = X_{\calS \setminus U}.
    \]
    But this follows from \Cref{subint} after recalling (\Cref{rem:matr}) that $\calS \setminus \{\alpha\}$ is linearly independent for any~$\alpha$.
\end{proof}

Finally, whenever $|T| \leq d-1$ the sets $\calS \setminus T \setminus \{\alpha\}$ are nonempty, and so we may therefore conclude using \Cref{fact:connected}:
\begin{corollary}\label{cor:connectedlinks}
    For all $\sigma \in \complex(i)$ with $i \le d-2$, $K_\sigma$ is connected.
\end{corollary}

\begin{remark}
The fact that $\complex$ and all of its links of dimension at most $d-2$ are connected is equivalent to saying that $\complex$ is \emph{strongly gallery connected} \cite[Rem.~2.1]{KO18}.
\end{remark}

\subsection{Expansion of links}\label{sec:explink}

\begin{definition}
    For $\alpha,\beta \in \Phi$ with $\alpha \neq -\beta$ we use the shorthand $\CC{\alpha}{\beta} = \CC{X_{\{\alpha,\beta\}} }{(X_\alpha,X_\beta)}$.
\end{definition}
It follows from \Cref{thm:linkstruct} that the link of every $(d-2)$-dimensional face in our complex~$\complex$ is isomorphic to $\CC{\alpha}{\beta}$ for distinct $\alpha,\beta \in \calS$.  The main goal of this section is to show that the bipartite skeleton graphs of these $\CC{\alpha}{\beta}$ are good expanders.  (For this we will not even need to recall our specific choice of~$\calS$.) Combined with \Cref{thm:trickle} and the connectivity result \Cref{cor:connectedlinks}, it follows that all links of $\complex$ are good expanders.

We begin with a simple observation:
\begin{proposition}
    For $\alpha \neq -\beta$, the (skeleton of) $\CC{\alpha}{\beta}$ is a $p^2$-regular bipartite (multi)graph.
\end{proposition}
\begin{proof}
    From \Cref{fact:link}, the link of a vertex in $X_{\alpha,\beta}/X_\beta$ is isomorphic to $\CC{X_\beta}{X_\alpha \cap X_\beta} = \CC{X_\beta}{1}$, where we used \Cref{subint}.  
    But this is equivalent to saying the neighborhood of a vertex in the skeleton is a set of size $|X_\beta| = p^2$ (recalling \Cref{subsize}).  The same consideration holds for vertices in $X_{\alpha,\beta}/X_\alpha$.
\end{proof}

The key idea we will use in understanding the expansion of  the links $\CC{\alpha}{\beta}$ will be to look at the graph-theoretic \emph{square}, $\CC{\alpha}{\beta}^2$, of (the skeleton of) $\CC{\alpha}{\beta}$.  Since $\CC{\alpha}{\beta}$ is connected and bipartite, we know that its random walk matrix has isolated ``trivial'' eigenvalues of~$\pm 1$, and all other eigenvalues are between $\pm \lambda_2(\CC{\alpha}{\beta})$.  Thus if we exclude from $\CC{\alpha}{\beta}^2$ the ``trivial'' eigenvalue~$1$, its maximum eigenvalue will be $\lambda_2(\CC{\alpha}{\beta})^2$, the square of what we wish to bound.  In fact, since $\CC{\alpha}{\beta}$ is bipartite, $\CC{\alpha}{\beta}^2$ will have two disconnected components corresponding to the two parts of $\CC{\alpha}{\beta}$.  It is a simple and well-known linear algebra fact that these two components have the same eigenvalues (possibly up to some eigenvalues of~$0$).
Hence it suffices for us to bound the eigenvalues of $\CC{\alpha}{\beta}^2$ on only \emph{one} of the two sides, $X_{\alpha,\beta}/X_\alpha$ or $X_{\alpha,\beta}/X_\beta$.  

As we will now show, whenever $\Phi \neq G_2$, at least one of these two sides is an abelian Cayley graph. (Interestingly, we do not know that both sides are.)  Thus we can understand the eigenvalues by elementary methods.  We discuss a potential approach to handling the~$G_2$ case in \Cref{sec:further}.

\begin{theorem}\label{thm:linkexpand}
        Let $\alpha, \beta \in \Phi \neq G_2$, with $\alpha \neq -\beta$. 
        Then the nontrivial eigenvalues of $\CC{\alpha}{\beta}^2$ are at most $2/p$; hence $\lambda_2(K_\sigma) \leq \sqrt{2/p}$ for every $\sigma \in \complex(d-2)$.
\end{theorem}
\begin{proof}
    When it is relevant, we will follow the convention of calling the shorter of the two roots~$\alpha$ and the longer~$\beta$.  
    Then, with foresight toward Case~3 below, we choose to study the $X_{\alpha,\beta}/X_\alpha$ side of $\CC{\alpha}{\beta}^2$.  
    
    By virtue of \Cref{substruct}, we can describe coset representatives for $X_{\alpha,\beta}/X_\alpha$ fairly simply; fixing an ordering for the roots in which $\alpha$ is last, we can take as coset representatives precisely those elements of the form
    \begin{equation}    \label{eqn:coset1}
        g = \prod\{x_{i \alpha + j \beta}(t_{ij}) : (i,j) \in \N \times \N \setminus \{(1,0)\},\ i\alpha+j\beta \in \Phi,\ \deg(t_{ij}) \leq i + j\}.
    \end{equation}
    Moreover, the $p^2$ neighbors (counted with multiplicity) of vertex $gX_\alpha$ in the squared (multi)graph $\CC{\alpha}{\beta}^2$ are the following cosets:
    \[
        \parens*{g \cdot x_{\alpha}(f_0) \cdot x_\beta(f_1)}X_{\alpha}, \quad \text{for } f_0, f_1 \in \F \text{ of degree at most $1$}.
    \]
    Via the commutator formula one sees that the associated coset representatives are
    \begin{equation}    \label{eqn:coset2}
        g \cdot x_{\alpha}(f_0) \cdot x_\beta(f_1) \cdot x_{\alpha}(-f_0)  = g \cdot x_\beta(f_1) \cdot [x_\beta(f_1),x_{\alpha}(-f_0)] = g \cdot x_\beta(f_1) \cdot \prod_{\substack{i,j \in \Z_+ \\ i \alpha + j \beta \in \Phi}} x_{i\alpha+j\beta}(C_{ij}^{\beta,\alpha} (-f_0)^i f_1^j).
    \end{equation}

    By \Cref{rem:2rootspan}, either $\alpha+\beta \notin \Phi$, or the root subsystem of $\Phi$ spanned by $\alpha$ and $\beta$ is one of $A_2$, $B_2$, or $G_2$. We will skip the case when $\alpha$ and $\beta$ span $G_2$, as it only arises when $\Phi = G_2$. Now as in the proof of \Cref{rk2gen}, we will do case analysis on the possible sets $R = \{i,j : i\alpha + j\beta \in \Phi\}$.

    \paragraph{Case 1: $R = \{\alpha,\beta\}$.}   If $\alpha+\beta \notin \Phi$, then $X_\alpha$ and $X_\beta$ commute by \Cref{fact:chevcom}, and it is easy to check that $\CC{\alpha}{\beta}$ is in fact the complete $p^2$-regular bipartite graph; hence the nontrivial eigenvalues of $\CC{\alpha}{\beta}^2$ are all~$0$.
    

    \paragraph{Case 2: $\{\alpha,\beta,\alpha+\beta\}$.}
    It was shown in~\cite{KO18}, and alternatively in~\cite[Corollary 5.6]{HS19}, that $\lambda_2(\CC{\alpha}{\beta}) = \sqrt{1/p}$; equivalently, the nontrivial eigenvalues of $\CC{\alpha}{\beta}$ are at most~$1/p$. Here we give a different proof of this fact, the strategy of which will be generalized in Case~3.
    
    From \Cref{eqn:coset1,eqn:coset2} we have that a typical coset representative 
    $
        g = x_\beta(t_{01})\cdot x_{\alpha+\beta}(t_{11})
    $
    is connected in $\CC{\alpha}{\beta}^2$ to the following coset representatives, for $f_0, f_1 \in \F$ of degree at most~$1$:
    \[
        x_\beta(t_{01})\cdot x_{\alpha+\beta}(t_{11})\cdot x_{\beta}(f_1) \cdot x_{\alpha+\beta}(-C_{11}^{\beta,\alpha} f_0f_1) = x_{\beta}(t_{01}+f_1) \cdot x_{\alpha+\beta}(t_{11} - C_{11}^{\beta,\alpha} f_0f_1).
    \]
    Reparameterizing with $f_2 = -C_{11}^{\beta,\alpha}f_0$ (and recalling $C_{11}^{\beta,\alpha} \neq 0$), it is evident that $\CC{\alpha}{\beta}^2$ is an abelian Cayley group, wherein each vertex is a pair $(\ell,q) $ with $\ell$~linear and $q$~quadratic, hence $(\ell,q) \cong \F_p^5$, and with edges involve adding a pair $(f_1,f_1f_2)$ for $f_1,f_2$ linear. With~$x$ denoting the field indeterminate, we can write $f_1 = a+bx$ and $f_2 = c + dx$; then the $X_{\alpha,\beta}/X_\alpha$ side of our graph~$\CC{\alpha}{\beta}^2$ may be identified as an abelian Cayley graph on~$\F_p^5$ with symmetric generating set 
    \[
        \{(a,b,ac,ad+bc,bd) : a,b,c,d \in \F_p^4\}.
    \]
    Then it is well known that the eigenvalues of this graph are given by the exponential sums
    \begin{align}
        &\mathrel{\phantom{=}} \E_{\ba,\bb,\bc,\bd \sim \F_p}\bracks*{\Exp_p(r_1\ba + r_2 \bb + r_3\ba\bc + r_4(\ba\bd+\bb\bc) + r_5\bb\bd)}  \nonumber \\
        &= \E_{\bc,\bd}\bracks*{
            \E_{\ba}\bracks*{\Exp_p(\ba \cdot h(\bc,\bd))}
            \E_{\bb}\bracks*{\Exp_p(\bb \cdot h'(\bc,\bd))}} \label{yo}
    \end{align}
    for $r_1, \dots, r_5 \in \F_p$, where $\Exp_p(z) = e^{2\pi i z/p}$, and
    \[
        h(c,d) = r_1 + r_3c + r_4d, \qquad h'(c,d) = r_2 + r_4c + r_5d.
    \]
    Notice whenever the outcome $\bc,\bd$ has $h(\bc,\bd) \neq 0$, the  quantity $\E_{\ba}\bracks*{\Exp_p(\ba \cdot h(\bc,\bd))}$ inside \Cref{yo} becomes~$0$.  On the other hand, if $h(\bc,\bd) = 0$ then this quantity is~$1$.  Similar considerations hold for~$h'$, and we conclude that the eigenvalue in \Cref{yo} is precisely
    \[
        \Pr_{\bc,\bd}\bracks*{h(\bc,\bd) = h'(\bc,\bd) = 0}.
    \]
    Of course if $r_1 = \cdots = r_5 = 0$ then $h,h'$ are formally~$0$ and the above is the trivial eigenvalue of~$1$. 
    But otherwise, at least one of $h,h'$ is nonzero --- say, $h$ --- and, being an  affine linear polynomial over~$\F_p$, it has $\Pr_{\bc,\bd}[h(\bc,\bd)] \leq 1/d$.  This shows that indeed the nontrivial eigenvalues of $\CC{\alpha}{\beta}^2$ are at most~$1/p$.

    \paragraph{Case 3: $R = \{\alpha,\beta,\alpha+\beta, 2\alpha+\beta\}$.}  
    As mentioned earlier, here we have named the shorter root~$\alpha$ and the longer root~$\beta$.
    From \Cref{eqn:coset1,eqn:coset2} we have that a typical coset representative 
    $
        g = x_\beta(t_{01})\cdot x_{\alpha+\beta}(t_{11}) \cdot x_{2\alpha+\beta}(t_{21})
    $
    is connected in $\CC{\alpha}{\beta}^2$ to the following coset representatives, for $f_0, f_1 \in \F$ of degree at most~$1$:
    \begin{multline*}
        x_\beta(t_{01})\cdot x_{\alpha+\beta}(t_{11})\cdot x_{2\alpha+\beta}(t_{21}) \cdot x_{\beta}(f_1) \cdot x_{\alpha+\beta}(-C_{11}^{\beta,\alpha} f_0f_1) \cdot x_{2\alpha+\beta}(C_{12}^{\beta,\alpha} f_0^2f_1) \\
        = x_\beta(t_{01} + f_1)\cdot x_{\alpha+\beta}(t_{11}-C_{11}^{\beta,\alpha} f_0f_1))\cdot x_{2\alpha+\beta}(t_{21}+C_{12}^{\beta,\alpha} f_0^2f_1).
    \end{multline*}
    (We remark that had we looked at the $X_{\alpha,\beta}/X_\beta$ side of $\CC{\alpha}{\beta}^2$, we would not have gotten all of the commutativity in the above calculation.)
    Reparameterizing again with $f_2 = -C_{11}^{\beta,\alpha} f_0$, this is
    \[
        x_\beta(t_{01} + f_1)\cdot x_{\alpha+\beta}(t_{11}+f_1f_2) \cdot x_{2\alpha+\beta}(t_{21}+C f_1f_2^2)
    \]
    for some constant $C \neq 0$ in~$\F_p$.
    Similar to Case~2, we see that this is an abelian Cayley graph on $\F_p^9$ with symmetric generating set
    \begin{align*}
        \{(a,b,ac,ad+bc,bd,Cac^2,C(bc^2+2acd),C(2bcd+ad^2),Cbd^2)
        : a,b,c,d \in \F_p\}.
        \end{align*}
        As before, the eigenvalues of this graph are given by
        \begin{multline*}
         \E_{\ba,\bb,\bc,\bd \in \F_p}\bigl[\Exp_p( r_1\ba+r_2\bb+r_3\ba\bc+r_4(\ba\bd+\bb\bc) + r_5\bb\bd +r_6C\ba\bc^2+ \\
        r_7C(\bb\bc^2+2\ba\bc\bd)+r_8C(2\bb\bc\bd+\ba\bd^2)+r_9C\bb\bd^2)\bigr]
        \end{multline*}
        \begin{equation} 
            \mathrel{=}  \E_{\bc,\bd}\bracks*{
            \E_{\ba}\bracks*{\Exp_p(\ba \cdot h(\bc,\bd))}
            \E_{\bb}\bracks*{\Exp_p(\bb \cdot h'(\bc,\bd))}}, \label{sum2}
        \end{equation}
        for all $r_1, \ldots, r_9 \in \F_p$, where
        \begin{align*}
        h(c,d) &= r_1+r_3c+r_4d+Cr_6c^2+2Cr_7cd+Cr_8d^2,\\
        h'(c,d) &= r_2+r_4c+r_5d+Cr_7c^2+2Cr_8cd+Cr_9d^2.
        \end{align*}
        The argument is now the same as in Case~2, except we reason that if $h$ is nonzero, then $\Pr_{\bc,\bd}[h(\bc,\bd)] \leq 2/p$ by Schwartz--Zippel, since now~$h$ is quadratic.
    \end{proof}

\begin{remark}\label{g2graphs}
When $\Phi = G_2$ two other graphs can arise as $\CC{\alpha}{\beta}$. The squares of these graphs are not Cayley graphs of abelian groups, and so the previous approach fails. For completeness we now give explicit description of the squared graphs $\CC{\alpha}{\beta}^2$ restricted to the vertices on the side $X_{\alpha,\beta}/X_\alpha$.

From \Cref{fig:root-eg} we see that if $\alpha, \beta \in G_2$ and $\alpha+\beta \in G_2$ then $\angle(\alpha,\beta) \in \{60^\circ, 120^\circ,150^\circ\}$. If $\angle(\alpha,\beta) = 60^\circ$ or if $\angle(\alpha,\beta) = 120^\circ$ and $\alpha$ and $\beta$ are long roots, the analysis is the same as in Case 2 of the previous proof. There are two remaining cases: (I)~$\alpha$ and $\beta$ are simple roots and $\angle(\alpha,\beta) = 150^\circ$ as in \Cref{fig:root-eg}; (II)~$\angle(\alpha,\beta) = 120^\circ$ and $\alpha$ and $\beta$ are short roots.

\paragraph{Case I: $\angle(\alpha,\beta) = 150^\circ$.} 
From \Cref{eqn:coset1}, a typical coset representative in $X_{\alpha,\beta}/X_\alpha$ is 
\[
    g=x_\beta(t_{01}) \cdot x_{\alpha+\beta}(t_{11})\cdot x_{2\alpha+\beta}(t_{21})\cdot x_{3\alpha+\beta}(t_{31}) \cdot x_{3\alpha+2\beta}(t_{32})
\]
with $\deg(t_{ij}) \leq i+j$. By \Cref{eqn:coset2}, the neighbors of this coset representative in $\CC{\alpha}{\beta}^2$ are parameterized by
\[g \cdot  x_\beta(f_1-t_{01})\cdot  [x_\beta(f_1-t_{01}),x_\alpha(-f_0)]\]
for all $f_0,f_1$ of degree at most 1. Using \Cref{rk2com}, one can show that this is the multigraph with vertices $(t_{01},t_{11},t_{21},t_{31},t_{32})$ whose neighbors are parameterized by
\[(f_1+t_{01}, -f_0f_1+t_{11},f_0^2f_1+t_{21},-f_0^3f_1+t_{31},-f_1(t_{31}+3t_{21}f_0+f_0^3f_1 )+t_{32}).\]

\paragraph{Case II: $\angle(\alpha,\beta) = 120^\circ$.}
A typical coset representative in $X_{\alpha,\beta}/X_\alpha$ is \[g=x_{\beta}(t_{01})\cdot x_{\alpha+\beta}(t_{11})\cdot x_{2\alpha+\beta}(t_{21})\cdot x_{\alpha+2\beta}(t_{12}),\]
where $\deg(t_{ij}) \le i+j$. The neighbors of this coset representative are parameterized by
\[g \cdot x_{\beta}(f_1-t_{01}) \cdot [x_{\beta}(f_1-t_{01}), x_\alpha(-f_0)] \]
for all $f_0,f_1$ of degree at most 1. Using \Cref{rk2com}, one can show that this is the multigraph with vertices $(t_{01},t_{11},t_{21},t_{12})$ whose neighbors are parameterized by
\[(f_1+t_{01},-2f_0f_1+t_{11},3f_0^2f_1+t_{21},3f_1(t_{11}+f_0f_1)+t_{12}).\]
\end{remark}
\section{Concluding}
Finally we can prove \Cref{thm:main}.
\begin{proof}[Proof of \Cref{thm:main}]
\quad
\begin{enumerate}
\item By \Cref{thm:order}, we have $|\Esc{\Phi}{\F}| = p^{\Theta(m)}$.
    By \Cref{subsize}, the subgroups $H_\alpha$ have size at most $p^{O(1)}$.
    (Here the $\Theta(\cdot)$ and $O(\cdot)$ depend only on~$\Phi$.)
    Hence there are $p^{\Theta(m)}$ total cosets, and the claim that $|\complex(0)| = p^{\Theta(m)}$ follows.
    Note that as $m$ increases by~$1$, the size of the complex grows by a constant factor~$p^{O(1)}$; thus we have the linear growth in size needed for a strongly explicit family, and the \emph{exact} number of vertices~$n$ can be computed efficiently in $\poly(m) = \polylog(n)$ time (by \Cref{thm:order}).  The resulting family is strongly explicit thanks to \Cref{thm:computing}: one can and construct all the group elements in $\Esc{\Phi}{\F}$ efficiently, one can identify the vertices (cosets) explicitly and naively by listing all their elements (recall each~$H_\alpha$ has constant size), and one can compute the complex's adjacency structure (e.g., list all maximal faces to which a given vertex belongs) thanks to the efficient ($\poly(m) = \polylog(n)$ time) group arithmetic from \Cref{thm:computing}.
\item Again, by \Cref{subsize} the subgroups $H_\alpha$ have size at most $p^{O(1)}$. The number of maximal faces containing a vertex is therefore at most $p^{O(1)\cdot (d+1)} = p^{O(1)}$.
\item This is \Cref{cor:connectedlinks}.
\item This is \Cref{thm:linkexpand}.
\item From \Cref{subint}, $\bigcap_{\alpha \in \calS} H_\alpha = \{1\}$. The claim then follows from \Cref{fact:type-preserving}.  In addition, by \Cref{obs:centerint} we have that $\bigcap_{\alpha \in \calS} Z H_\alpha = Z$, and so in the quotient group $\G{\Phi}{\F}$ the subgroups $H_i Z/Z$ intersect trivially. Hence we also get a simply-transitive action for the adjoint Chevalley groups.\qedhere
\end{enumerate}

\end{proof}
\subsection{Further questions}  \label{sec:further}
As mentioned in \Cref{g2graphs}, we do not know if the $2$-dimensional complexes obtained from our construction for $\Esc{G_2}{\F}$ yield HDX families, either in the case  $\calS = \{\alpha, \beta, -\alpha-\beta\}$ (as we selected in \Cref{def:rootchoice}) or in the alternative case $\calS = \{\alpha,\alpha+\beta,-2\alpha-\beta\}$ mentioned in \Cref{rem:alt}. We note that this latter case with $\calS = \{\alpha,\alpha+\beta,-2\alpha-\beta\}$ is particularly appealing, as all vertex links are isomorphic (i.e., the $1$-skeleton is a \emph{graph of constant link}). Additionally, to the best of our knowledge none of the links in this case arise in previous HDX constructions. This is in contrast to our other constructions, which always contain some links isomorphic to those studied in \cite{KO18}.

One approach to prove the expansion of these links is to count the number of closed walks of some fixed length $k$ in one side of their square, which (by the trace method) equals the sum of the $k$th powers of the eigenvalues of their adjacency matrices. By \Cref{g2graphs}, the number of length-$k$ paths starting and ending at a fixed vertex on the side $X_{\alpha,\beta}/X_\alpha$ equals the number of solutions to the following systems of equations, corresponding to the first and second cases in \Cref{g2graphs}, respectively:
\begin{align*}
0 &= \sum_{i=1}^k g_i = \sum_{i=1}^k f_ig_i = \sum_{i=1}^k f_i^2g_i = \sum_{i=1}^k f_i^3g_i=\sum_{i=1}^k -g_i(f_i^3g_i+\sum_{j=1}^{i-1} (f_j^3g_j+3f_j^2g_jf_i)),\\
0 &=\sum_{i=1}^k g_i = \sum_{i=1}^k f_ig_i = \sum_{i=1}^k f_i^2g_i = \sum_{i=1}^k g_i(f_ig_i - 2\sum_{j=1}^{i-1} f_jg_j).
\end{align*}
Here $f_i$ and $g_i$ are linear polynomials in $\F_p[x]$. The graphs corresponding to the first and second systems have $p^n$ vertices, where $n=20$ and $n=13$, respectively. Therefore if for some particular~$k$ one could bound the number of solutions to either of these by, say, $p^{4k-n} + p^{3.99k}$, expansion of the corresponding complexes would follow. To show this it would suffice to show that the varieties defined over $\mathbb{C}$ by these systems are irreducible and of dimension at most $4k-20$ in the first case, or $4k-13$ in the second case, for some~$k$. This seems potentially tractable for a computer algebra system.

Finally, we have shown that the untwisted groups of Lie type act (simply) transitively on the maximal faces of certain HDXs. Are there more ``combinatorial" families of groups --- perhaps the symmetric group or the generalized symmetric groups --- which admit transitive actions on HDXs?
\nocite{Sch10}
\bibliography{refs}
\end{document}